\documentclass[12pt]{amsart}

\usepackage{vmargin}
\usepackage[dvips]{graphicx}
\usepackage{color}
\input{amssym.def}
\input{amssym}
\usepackage{a4wide}
\usepackage{supertabular}
\usepackage{array}
\usepackage{multicol}
\usepackage{xcolor}
\usepackage{graphbox}

\color{black}

\usepackage{hyperref}
\setmarginsrb{3cm}{2cm}{3cm}{3cm}{75pt}{20pt}{20pt}{30mm}
\def\R{\mathbb R}

\def\<{\langle}
\def\>{\rangle}

\def\sm{\smallskip}
\def\bul{$\bullet$\ }

\def\begeq{\begin{equation}}
\def\endeq{\end{equation}}
\def\begar{\begin{eqnarray}}
\def\endar{\end{eqnarray}}
\def\begar*{\begin{eqnarray*}}
\def\endar*{\end{eqnarray*}}
\def\begal{\begin{align}}
\def\endal{\end{align}}
\def\begal*{\begin{align*}}
\def\endal*{\end{align*}}

\newtheorem{Thm}{Theorem}

\newtheorem{Prop}[Thm]{Proposition}

\numberwithin{equation}{section}
\numberwithin{Thm}{section}

\theoremstyle{definition}

\newtheorem{Rk}[Thm]{Remark}

\theoremstyle{remark}

\newtheorem*{Thm*}{Theorem}
\newtheorem*{Lem*}{Lemma}
\newtheorem*{Conj*}{Conjecture}
\newtheorem*{Cor*}{Corollary}
\newtheorem*{Def*}{Definition}
\newtheorem*{Prop*}{Proposition}
\newtheorem*{Exo*}{Exercise}
\newtheorem*{Exs*}{Examples}
\newtheorem*{Ex*}{Example}
\newtheorem*{Rk*}{Remark}
\newtheorem*{Rks*}{Remarks}

\def\signar{\bigskip \begin{center} {\sc \'Alvaro I. Riquelme\par\vspace{3mm}
Queen's University \par
The Robert M. Buchan Department of Mining\par
CANADA\par\vspace{3mm}
e-mail:} \tt{alvaro.riquelme@queensu.ca} \end{center}}

\def\signjo{\bigskip \begin{center} {\sc Juli\'an M. Ortiz\par\vspace{3mm}
			Queen's University \par
			The Robert M. Buchan Department of Mining\par
			CANADA\par\vspace{3mm}
			e-mail:} \tt{ julian.ortiz@queensu.ca} \end{center}}

\begin{document}

\title[A general approach to the assessment of uncertainty in volumes]{A general approach to the assessment of uncertainty in volumes by using the multi-Gaussian model \\[1mm]
  {\today}}

\vspace*{-8mm}

\author{AI Riquelme}
\author{JM Ortiz}

\begin{abstract}
The goal of this research is to derive an approach to assess uncertainty in an arbitrary volume conditioned by sampling data, without using geostatistical simulation. We have accomplished this goal by deriving an numerical tool suitable for any probabilistic distribution of the sample data. For this, we have worked with an extension of the traditional multi-Gaussian model, allowing us to obtain a formulation that makes explicit the dependence of the uncertainty in the arbitrary volume from the grades within the volume, the spatial correlation of the data and the conditioning values. A Kriging of the Gaussian values is the only requirement to obtain not only conditional local means and variances but also the complete local distributions at any support, in an easy and straightforward way.
\end{abstract}

\maketitle

\vspace*{-13mm}

\tableofcontents

\vspace*{-11mm}


\bigskip

The Multi-Gaussian model is a very simple non linear technique used in the context of spatial estimation of variables, introduced by Matheron in the 70's \cite{Disjonctif,Matheronnato,matheron1974fonctions}. It solves many classic problems regarding spatial estimation, when the linear estimation does not provide a good solution.

One of the most successful applications of the model has relationship with the simulation of ore deposits, producing different plausible scenarios, conditioned to a set of samples (i.e. matching the data in their positions and values), reproducing their spatial variability. These scenarios are post-processed, using non-linear transfer functions (usually related with the application of cutoff grades in open-pit design, mine planning and exploration) in order to assess the impact of the geological variability on mineral resources and reserves.

The multi-Gaussian model also helps obtaining  estimators of local conditional expectations and variances, at point and block support, including arbitrary volumes \cite{Matheronnato,Ortiz2005,Mar2,Rendu1980}. The discrete Gaussian model \cite{Forecasting,Emery2005a} can also be used to provide the change of support distribution locally, however it does not consider the proportional effect or the blending when an arbitrary number of different zones is considered.

The goal of this research is related with this last problem. Simulations allow to address straightforwardly the probability distribution of variables over any volume, just by changing the support of the outcomes over different scenarios. However, its result depend on the number of realizations and, therefore, provides an inexact result.

We present next the steps to close this gap by analyzing the inference of the mean estimator, then the variances and finally the complete conditional distributions, at point support and when change of support is considered. This is done considering both the conditioning effect of sample locations and the proportional effect, without resorting to the use of simulations.

The plan of the paper is as follows. In Section \ref{sec:intro} we provide an introduction related with the history and the range of validity of the model. Then in Section \ref{sec:main}, we comment on our procedure to obtain the variance for any support. In Section \ref{sec:MG}, we formally establish the multi-Gaussian model. After the statement, it is rather easy to derive point $n$-variate probability distributions. Some examples are presented at the end of that section, of the application of the model and its comparison with standard simulation results. In Section \ref{sec:ChangeSup}, we propose our methodology for change of support with an example, providing a general and simple approach. Finally, Section \ref{sec:Her} provides an introduction and classic results of Hermite Polynomials. Hermite Polynomials are highly suitable to be incorporated into the Gaussian model. This section serves as a compilation of proofs to the reader and can be read independently of the rest. An elegant result of the application of Hermite Polynomials is presented in Section \ref{sec:Application}. Final comments of these results and future work are provided in Section \ref{sec:Conclusion}.

\section{Introduction to the Multi-Gaussian Model}
\label{sec:intro}

\subsection{Historic perspective}

The multi-Gaussian model was firstly introduced in Geostatistics by Matheron in the 70's \cite{matheron1974fonctions} under the name of Disjunctive Kriging, as a way of directly estimating non-linear functions of stationary grade values. Under this multivariate Gaussian framework, the derivation of conditional distributions, and hence of estimates for recovery functions (grade and tonnage above a cutoff), is particularly straightforward.

Several publications, since then, can be found in the literature which present the theory and applications of disjunctive kriging \cite{marechal1,journel1978mining,Matheronnato}. The problem considered is the  estimation of the attribute $Z$ in a volume $\boldsymbol{V}$ (spatially discretized in an regular way by points $\textbf{u}_{i}, i \in \{1,\dots,|\boldsymbol{V}|\}$), from $n$ point samples with values $z(\textbf{u}_{\alpha}),\alpha \in \{1,\dots,n\}$, located at points $\textbf{u}_{\alpha},\alpha \in \{1,\dots,n\}$, inside or outside the volume $\boldsymbol{V}$. If $Z^*_{\boldsymbol{V}}$ is the estimator for the unknown value of the volume, it is usually desirable that this estimator presents properties such as unbiasedness, conditional unbiasedness and minimum error variance. Different estimators possess these properties, such as kriging, suitable if the sample values are normally distributed, and logarithmic kriging, when lognormality prevails. However, estimators derived from the multi-Gaussian model have the advantage of not only having these properties, but also of being distribution free, i.e., the model can be used whatever the multivariate probability distribution of the sample values \cite{Rendu1980}.

\subsection{Notion}
The idea of the multivariate Gaussian approach is to transform the initial random function $Z(\textbf{u})$, into a random function $Y(\textbf{u})$ with a standard Gaussian univariate distribution. Then, under the working hypothesis that $Y(\textbf{u})$ is spatially multi-Gaussian, the conditional distributions of $Y(\textbf{u})$ are determined. Derivation of conditional expectation in a non-linear way, and estimates of recovery functions, are then obtained through inverse transformation from these $Y(\textbf{u})$ conditional distributions.

In practice, the non-linear estimator provided by this procedure differs from the true conditional expectation. However, and under the condition that the univariate distribution of $Z(\textbf{u})$ is well known, practice has shown that the mean estimator of this non-linear technique, $[Z^*(\textbf{u})]$, is generally better than the estimators provided by the direct linear kriging processes applied to the initial data (see \cite{rendu1979normal,krige1976some,krige1951statistical}).
Some caution, however, should be taken with very continuous variables, which may result in the overestimation of areas around high-valued samples and, consequently, in a dangerous bias if used for selecting mining blocks \cite{Rivo}.

\subsection{Range of validity}
The multi-Gaussian approach is very convenient: the inference of the conditional distribution function (cdf) reduces to solving a simple kriging system at location $\textbf{u}$. The trade-off cost is the assumption that data follow a multi-Gaussian distribution, which implies first that the one-point distribution of data (sample histogram) is normal.

Many variables in the earth sciences show an asymmetric distribution with a few very large values (positive skewness). Thus, the multi-Gaussian approach starts with an identification of the standard distribution and involves the ``normalizing'' of the skewed sample histogram. The common approach for this normalization is to apply a normal score transformation. This procedure will be described in detail in the following sections.

However, the normality of the one-point cdf is a necessary but not sufficient condition to ensure that the random function model is multivariate Gaussian. Hence one should check whether the normalized data are also reasonably bivariate Gaussian. There are several ways to check that the two-point distribution of data $y(\textbf{u}_{\alpha}),\alpha \in \{1,\dots,n(\textbf{u})\}$, is bivariate normal. For a detailed
exposition, and many explanations, the reader can check \cite{goovaerts1997geostatistics,Emery2005b}.

\subsection{Conceptual problems}
The normality of one-point and two-point cdfs are both necessary but not sufficient conditions to ensure that a multivariate Gaussian random function model is appropriate for modeling the spatial distribution of normal score data. One should also check the normality of the three-point until the $N$-point cdfs. One type of check consist of comparing experimental multiple-point frequencies with their theoretical Gaussian values through the comparison with analytical expressions. Although such an expression has been established, the main difficulty resides in the inference of such experimental frequencies. For example, the inference of a three-point frequency such that $$P(Y(\textbf{u})\leq y_p,Y(\textbf{u}')\leq y_p,Y(\textbf{u}'')\leq y_p) $$ requires the availability of a series of triplet values with the same geometric configuration as the triplet $(\textbf{u},\textbf{u},\textbf{u}'')$.

Non-regular drilling grids and data sparsity prevent us from computing sample statistics involving more than two locations at a time. Therefore, in practice, because most of testing becomes tedious, if $h$-scatter plots do not invalidate the bi-Gaussian assumption, the multi-Gaussian formalism is adopted.

\section{Main result}
\label{sec:main}

Our principal concern, when conceiving a tool capable of the modeling spatial uncertainty at any support, is to take into consideration, on one hand, the spatial location of the samples in a way that the uncertainty is reduced when we approach to the samples and, on the other hand, to reproduce the heteroscedasticity that we find in real deposits, i.e., to model the zones of higher values with higher variance. This heteroscedastic behavior is commonly referred to as the proportional effect.

The main development, presented in what follows, has a rather simply but powerful origin, and is based in the following variance calculation for an average of values within a volume:

\begin{align}\label{covav}
\mathbb{V}ar(Z_{\mathit{\boldsymbol{V_{k}}}}|Z(\textbf{u}_{\alpha})=z(\textbf{u}_{\alpha}),\alpha \in \{1,\dots,n(\textbf{k})\})=\displaystyle \frac{1}{|\boldsymbol{V_k}|^2}\sum_{\substack{i:\textbf{u}_{i} \in\boldsymbol{V_k}\\j:\textbf{u}_{j}\in\boldsymbol{V_k}}}\mathbb{C}ov(Z^*(\textbf{u}_{i}),Z^*(\textbf{u}_{j})),
\end{align}
with $|\boldsymbol{V_k}|$ the amounts of points to average in  $\boldsymbol{V_k}$. This result allows us to compute the variance of a volume in a simple way, by taking an ``average'' of the covariance between pairs of points within the volume.

The key step, that follows to this result, is to consider this covariance between the pairs of points $\mathbb{C}ov(Z^*(\textbf{u}_{i}),Z^*(\textbf{u}_{j}))$, without the assumption of stationarity in the original values, but in their respective Gaussian values $\mathbb{C}ov(Y^*(\textbf{u}_{i}),Y^*(\textbf{u}_{j}))$. Then, the covariance in the original values is retrieved as a non-linear function of the covariance  $\mathbb{C}ov(Y^*(\textbf{u}_{i}),Y^*(\textbf{u}_{j}))=\mathbb{C}ov(\textbf{u}_{i}-\textbf{u}_{j})$ and also a function of the conditioning sample values. This non-linear function will be derived in a straightforward way once the multi-Gaussian model is presented.

\section{The Multi-Gaussian Model}
\label{sec:MG}

In the multi-Gaussian model, the attribute under study is viewed as a realization of a
random function $Z(\textbf{u})$ that can be transformed into a Gaussian random field $Y(\textbf{u})$ with
zero mean and unit variance, in a certain position $\textbf{u}$. This transformation is a quantile transformation (\textbf{u} will be implicit most of the time hereafter):
\begin{equation}
F_{Z}(z)=F_{Y}(y),\label{qt}
\end{equation}
where $F_{Z}(z)=P(Z\leq z)$ and $F_{Y}(y)=P(Y\leq y)$, commonly written as $G(y)$ (Figure \ref{GAn}). Hence,
$$Z(\textbf{u})=\phi[Y(\textbf{u})]$$
or just $z=\phi(y)$. Lets define  the cumulative density function (cdf) of a Gaussian random function  with mean $\mu$ and variance $\sigma^2$, as $G_{\mu}^{\sigma^2}$, except for the standard case $G_0^1$, written just as $G$. Then $\phi = F^{-1}_Z\circ G$, in shortened notation, is the anamorphosis function.

\begin{figure}[h]
	\begin{center}
		\includegraphics[width=14cm]{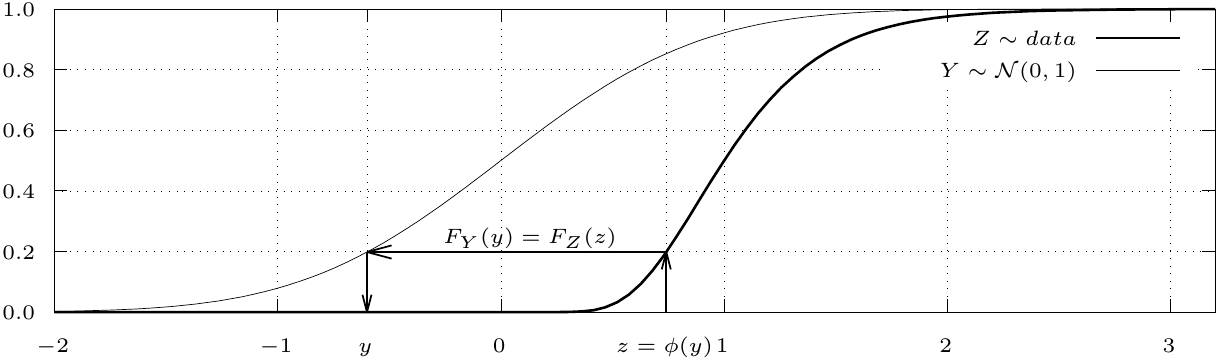}
	\end{center}
	\caption{Graphical representation of the Gaussian anamorphosis.}
	\label{GAn}
\end{figure}

We notice that the quantile transformation (\ref{qt}), written in terms of the probability density and the anamorphosis function, is equivalent to:
$$\int_{-\infty}^{z} f_Z(\textbf{z})d\textbf{z} = \int_{-\infty}^{\phi(y)} f_Z(\textbf{z})d\textbf{z} = \int_{-\infty}^{y} f_Y(\textbf{y})d\textbf{y},$$
leading to:
$$f_Z(z) = \frac{f_Y(y)}{\phi^{'}(y)}=\frac{g(y)}{\phi^{'}(y)},$$
where $\displaystyle \phi^{'}(y)=\frac{d\phi(y)}{dy}$ and $f_Y(y)$ is the probability density function (pdf) of a standard Gaussian distribution, commonly written as $g(y)$. We will define, in general, $g_{\mu}^{\sigma^2}$ as the pdf of a Gaussian random function with mean $\mu$ and variance $\sigma^2$ (except for the standard case):
\begin{equation}
g(y)=\frac{1}{\sqrt{2\pi}}\cdot e^{-\frac{y^2}{2}}, \qquad g_{\mu}^{\sigma^2}(y)=\frac{1}{\sigma \sqrt{2\pi}}\cdot e^{-\frac{(y-\mu)^2}{2\sigma^2}}.
\label{Gauss} 
\end{equation}
Knowing the anamorphosis of a variable $Z(\textbf{u})$ is equivalent to knowing its distribution.

The multi-Gaussian assumption states that the distribution of any value  $Y(\textbf{u})$
conditioned by the sampled values is still Gaussian, with a mean and a variance equal to its simple kriging estimate $y^*_{SK}(\textbf{u})$ and simple kriging variance $\sigma_{SK}^2(\textbf{u})$
respectively. Hence, the conditional cumulative probability function (ccdf) is
$$F_{Y|data}(y)=G\Bigg(\frac{y-y^*_{SK}}{\sigma_{SK}}\Bigg),$$
where $F_{Y|data}(y)=P(Y\leq y|data)=G(y|data)$, and where ``data'' represents the conditioning data $Y(\textbf{u}_{\alpha})=y(\textbf{u}_{\alpha}),\alpha \in \{1,\dots,n(\textbf{k})\}$. The conditional cdf of the original variable is then retrieved as:
\begin{align*}
		F_{Z|data}(z)&=F_{Y|data}(y)\\
		&=\displaystyle G\Bigg(\frac{y-y^*_{SK}}{\sigma_{SK}}\Bigg)=G_{y^*_{SK}}^{\sigma_{SK}^2}(y),	
\end{align*}
where $F_{Z|data}(z)=P(Z\leq z|data)$. This can be rewritten in terms of the probability densities:
\begin{center}
	\begin{tabular}{c c l}
		$\displaystyle \int_{-\infty}^{\phi(y)} f_{Z|data}(\textbf{z})d\textbf{z}$&=&$\displaystyle \frac{1}{\sigma_{SK}}\cdot \int_{-\infty}^{y} g\Bigg(\frac{\textbf{y}-y^*_{SK}}{\sigma_{SK}}\Bigg)d\textbf{y}$
		\\
		&=&$\displaystyle \int_{-\infty}^{y} g_{y^*_{SK}}^{\sigma_{SK}^2}(\textbf{y})d\textbf{y},$\\
	\end{tabular}
\end{center}
from where follows:
\begin{equation}
f_{Z|data}(z)\cdot\phi^{'}(y) =g_{y^*_{SK}}^{\sigma_{SK}^2}(y).\label{Distr}
\end{equation}
From here we can obtain the conditional density function without incurring in any quantile sampling of the posterior Gaussian distribution to construct the local distribution in the raw values:
\begin{equation}\label{back}
f_{Z|data}(z)=\frac{1}{\phi^{'}(y)}\cdot g_{y^*_{SK}}^{\sigma_{SK}^2}(y)=\frac{1}{\phi^{'}(\phi^{-1}(z))}\cdot g_{y^*_{SK}}^{\sigma_{SK}^2}(\phi^{-1}(z)).
\end{equation}
We will say that $Z|data \sim \boldsymbol{\phi}(y^*_{SK},\sigma_{SK}^2)$. Let us define $Z^*:= Z|data$ hereafter.

Multiplying both sides by $\phi(y)$ in (\ref{Distr}) and integrating with respect to $y$ lead us to:
\begin{align}
\displaystyle\int_{-\infty}^{\phi(y)} \textbf{z}\cdot f_{Z|data}(\textbf{z})d\textbf{z}
&=\displaystyle \int_{-\infty}^{y} \phi(\textbf{y}) \cdot g_{y^*_{SK}}^{\sigma_{SK}^2} (\textbf{y})d\textbf{y}\nonumber\\	
&=\displaystyle\frac{1}{\sigma_{SK}}\cdot\int_{-\infty}^{y} \phi(\textbf{y}) \cdot g\Bigg(\frac{\textbf{y}-y^*_{SK}}{\sigma_{SK}}\Bigg)d\textbf{y}\nonumber\\
&=\displaystyle\int_{-\infty}^{y}\phi(\sigma_{SK}\cdot\textbf{y}+y^*_{SK}) \cdot g(\textbf{y})d\textbf{y}. \label{exp}
\end{align}

By taking $y\to\infty$, we obtain $[Z^*(\textbf{u})]:=\mathbb{E}[Z(\textbf{u})|data]$:
\begin{equation}
[Z^*(\textbf{u})]=\int_{-\infty}^{\infty} \phi(\textbf{y}) \cdot g_{y^*_{SK}}^{\sigma_{SK}^2} (\textbf{y})d\textbf{y}.\label{imp}
\end{equation}
Let us insist that $[Y^*(\textbf{u})]=y^*_{SK}$, however, from (\ref{imp}), it is clear that $[Z^*(\textbf{u})]$ is not just equal to $\phi([Y^*(\textbf{u})])$, which usually can lead to confusion and wrong results.

In the same manner, it is easy to obtain any moment (if they exist) under the multi-Gaussian model:
$$\mathbb{E}[Z^{n}(\textbf{u})|data]=\int_{-\infty}^{\infty} \phi^{n}(\textbf{y}) \cdot g_{y^*_{SK}}^{\sigma_{SK}^2} (\textbf{y})d\textbf{y},$$
leading to an expression to compute the variance in a general way for the raw variable:
\begin{align}\label{variance}
\mathbb{V}ar[Z(\textbf{u})|data]&=\mathbb{E}[Z^{2}(\textbf{u})|data]-\mathbb{E}[Z(\textbf{u})|data]^{2}\\
&=\int_{-\infty}^{\infty} \phi^{2}(\textbf{y}) \cdot g_{y^*_{SK}}^{\sigma_{SK}^2} (\textbf{y})d\textbf{y}\nonumber-[Z^*(\textbf{u})]^2.
\end{align}

We can extend the model into a $n$-variate model in the following way: the joint cumulative function of $n$ random variables $Y(\textbf{u}_1),\dots,Y(\textbf{u}_n)$ (related to $Z(\textbf{u}_1),\dots,Z(\textbf{u}_n)$) at locations $\textbf{u}_1,\dots,\textbf{u}_n$, respectively, conditioned by the sampled values, follows a multivariate Gaussian distribution $\mathcal{N}( \boldsymbol{\mu}_{SK},\boldsymbol{\Sigma}_{SK})$ with mean vector $\boldsymbol{\mu}= (y^*_{SK}(\textbf{u}_1),\dots,y^*_{SK}(\textbf{u}_n))' $ and covariance matrix equal to $\boldsymbol{\Sigma}=(\boldsymbol{\Sigma})_{ij}$, $i,j \in \{1,\dots,n\}$, such that $ (\boldsymbol{\Sigma})_{ij} = (\boldsymbol{\Sigma})_{ji}= \sigma_{SK}(\textbf{u}_i,\textbf{u}_j)= \sigma_{SK}(\textbf{u}_j,\textbf{u}_i)$ and $ (\boldsymbol{\Sigma})_{ii} = \sigma_{SK}^2(\textbf{u}_i)$, $i \in \{1,\dots,n\}$.

In the bi-variate case the model looks as follows:

\begin{center}
	\begin{tabular}{c c l}
		$F_{Z(\textbf{u}_0),Z(\textbf{u}_1)|data}(z_0,z_1)$
		&=&$P(Z(\textbf{u}_0)\leq z_0,Z(\textbf{u}_1)\leq z_1|data)$\\	
		&&\\
		&=&$P(Y(\textbf{u}_0)\leq y_0,Y(\textbf{u}_1)\leq y_1|data)$\\
		&&\\
		&=&$F_{Y(\textbf{u}_0),Y(\textbf{u}_1)|data}(y_0,y_1)$\\
		&&\\
		&=&$G_{\boldsymbol{\mu}_{SK}}^{\mathbf{\Sigma}_{SK}}(y_0,y_1),$\\		
	\end{tabular}
\end{center}
with:
\begin{center}
	\begin{tabular}{c c l}
		$\boldsymbol{\mu}_{SK}$&=&$
		\begin{pmatrix}
		y^*_{SK}(\textbf{u}_0) \\
		y^*_{SK}(\textbf{u}_1) \\
		\end{pmatrix}
		$\\
		&&\\
		$\mathbf{\Sigma}_{SK}$&=&$
		\begin{pmatrix}
		\mathbb{E}\{(Y^*(\textbf{u}_{0})-Y(\textbf{u}_{0}))^2\}& \mathbb{E}\{(Y^*(\textbf{u}_{0})-Y(\textbf{u}_{0}))\cdot(Y^*(\textbf{u}_{1})-Y(\textbf{u}_{1}))\}\\
		\mathbb{E}\{(Y^*(\textbf{u}_{0})-Y(\textbf{u}_{0}))\cdot(Y^*(\textbf{u}_{1})-Y(\textbf{u}_{1}))\}&\mathbb{E}\{(Y^*(\textbf{u}_{1})-Y(\textbf{u}_{1}))^2\} \\
		\end{pmatrix}
		$ \\
		&&\\
		&=&$
		\begin{pmatrix}
		\sigma^2-\textbf{k}_{0}^T \cdot\textbf{C}^{-1} \cdot \textbf{k}_{0} & \sigma_{\textbf{u}_{0},\textbf{u}_{1}}-\textbf{k}_{0}^T \cdot\textbf{C}^{-1} \cdot \textbf{k}_{1}\\
		\sigma_{\textbf{u}_{1},\textbf{u}_{0}}-\textbf{k}_{1}^T \cdot\textbf{C}^{-1} \cdot \textbf{k}_{0}&\sigma^2-\textbf{k}_{1}^T \cdot\textbf{C}^{-1} \cdot \textbf{k}_{1} \\
		\end{pmatrix}
		$\\
		&&\\
		&=&$
		\begin{pmatrix}
		\sigma_{SK}^2(\textbf{u}_0) & \sigma_{SK}(\textbf{u}_0,\textbf{u}_1)\\
		\sigma_{SK}(\textbf{u}_1,\textbf{u}_0)&\sigma_{SK}^2(\textbf{u}_1) \\
		\end{pmatrix}.
		$ \\
	\end{tabular}
\end{center}

This can be rewritten in terms of probability densities:
\begin{center}
	\begin{tabular}{c c l}
		$\displaystyle\int_{-\infty}^{z_0} \int_{-\infty}^{z_1} f_{Z_0,Z_1|data} (\mathbf{z_0},\mathbf{z_1}) d\mathbf{z_0} d\mathbf{z_1}$
		&=&$\displaystyle\int_{-\infty}^{\phi(y_0)} \int_{-\infty}^{\phi(y_1)} f_{Z_0,Z_1|data} (\mathbf{z_0},\mathbf{z_1}) d\mathbf{z_0} d\mathbf{z_1}$\\	
		&&\\
		&=&$\displaystyle\int_{-\infty}^{y_0} \int_{-\infty}^{y_1} g_{\boldsymbol{\mu}_{SK}}^{\mathbf{\Sigma}_{SK}}(\mathbf{y_0},\mathbf{y_1}) d\mathbf{y_0} d\mathbf{y_1},$\\
	\end{tabular}
\end{center}
which leads to the following equality:
\begin{equation}
f_{Z_0,Z_1|data} (\phi(y_0),\phi(y_1))\cdot \phi'(y_0) \cdot \phi'(y_1) = g_{\boldsymbol{\mu}_{SK}}^{\mathbf{\Sigma}_{SK}}(y_0,y_1). 
\label{e1}
\end{equation}
Hence, we can obtain the probability distribution by:
\begin{equation}
f_{Z_0,Z_1|data} ({z_0},{z_1})= \frac{1}{\phi'(y_0)} \cdot \frac{1}{\phi'(y_1)} \cdot g_{\boldsymbol{\mu}_{SK}}^{\mathbf{\Sigma}_{SK}}(y_0,y_1).\label{e2}
\end{equation}

Multiplying both sides by $\phi(y_0)\cdot\phi(y_1)$ in \ref{e1} and integrating with respect to $y_0$ and $y_1$ we obtain:
$$\int_{-\infty}^{\phi(y_0)} \int_{-\infty}^{\phi(y_1)} \mathbf{z_0}\mathbf{z_1} f_{Z_0,Z_1|data} (\mathbf{z_0},\mathbf{z_1}) d\mathbf{z_0} d\mathbf{z_1}= \int_{-\infty}^{y_0} \int_{-\infty}^{y_1} \phi(\mathbf{y_0})\phi(\mathbf{y_1}) g_{\boldsymbol{\mu}_{SK}}^{\mathbf{\Sigma}_{SK}}(\mathbf{y_0},\mathbf{y_1}) d\mathbf{y_0} d\mathbf{y_1}.$$

By taking $y_0$, $y_1\to\infty$, we obtain $\mathbb{E}[Z(\textbf{u}_0)^*Z(\textbf{u}_1)^*]:=\mathbb{E}[Z(\textbf{u}_0)Z(\textbf{u}_1)|data]$:
\begin{align*} 
\mathbb{E}[Z(\textbf{u}_0)^*Z(\textbf{u}_1)^*] =\int_{-\infty}^{\infty} \int_{-\infty}^{\infty} \phi(\mathbf{y_0})\phi(\mathbf{y_1}) g_{\boldsymbol{\mu}_{SK}}^{\mathbf{\Sigma}_{SK}}(\mathbf{y_0},\mathbf{y_1}) d\mathbf{y_0} d\mathbf{y_1}.
\end{align*}
Combining the previous results, we obtain a general expression for the covariance under the multi-Gaussian model. We can see explicitly the covariance as a measure of stochastic dependence between the locations  $\textbf{u}_0$ and $\textbf{u}_1$:
\begin{align} \label{cov}
	\mathbb{C}ov(Z^*(\textbf{u}_0),Z^*(\textbf{u}_1))&=\mathbb{E}[Z(\textbf{u}_0)^*Z(\textbf{u}_1)^*]-[Z^*(\textbf{u}_0)][Z^*(\textbf{u}_1)]\\
	&=\displaystyle \int_{-\infty}^{\infty} \int_{-\infty}^{\infty} \phi(\mathbf{y_0})\phi(\mathbf{y_1}) g_{\boldsymbol{\mu}_{SK}}^{\mathbf{\Sigma}_{SK}}(\mathbf{y_0},\mathbf{y_1}) d\mathbf{y_0} d\mathbf{y_1}\nonumber\\
	&\displaystyle -\int_{-\infty}^{\infty} \phi(\textbf{y}) \cdot g_{y^*_{SK}(\textbf{u}_0)}^{\sigma_{SK}^2(\textbf{u}_0)} (\textbf{y})d\textbf{y}\cdot \int_{-\infty}^{\infty} \phi(\textbf{y}) \cdot g_{y^*_{SK}(\textbf{u}_1)}^{\sigma_{SK}^2(\textbf{u}_1)} (\textbf{y})d\textbf{y}. \nonumber	
\end{align}
Then, by plugging in this result in Equation \ref{covav}, we obtain a general answer for the assessment of the variance in the volume $\boldsymbol{V}$ under the multi-Gaussian framework.

In the n-variate case the mean vector and covariance matrix of the model looks as follows:
\begin{center}
	\begin{tabular}{c c l}
		$\boldsymbol{\mu}_{SK}$&=&$
		\begin{pmatrix}
		y^*_{SK}(\textbf{u}_1) \\
		\vdots \\
		y^*_{SK}(\textbf{u}_n) \\
		\end{pmatrix}
		$\\
		&&\\
		$\mathbf{\Sigma}_{SK}$&=&$
		\begin{pmatrix}
		\sigma_{SK}^2(\textbf{u}_1) &\cdots & \sigma_{SK}(\textbf{u}_1,\textbf{u}_n)\\
		\vdots&\ddots & \vdots\\
		\sigma_{SK}(\textbf{u}_n,\textbf{u}_1)& \cdots&\sigma_{SK}^2(\textbf{u}_n) \\
		\end{pmatrix},
		$ \\
		
	\end{tabular}
\end{center}
and the ccdf is extended by the link with a $n$-variate Gaussian cdf:
\begin{center}
	\begin{tabular}{c c l}
		$F_{Z(\textbf{u}_1),\dots,Z(\textbf{u}_n)|data}(z_1,\dots,z_n)$
		&=&$P(Z(\textbf{u}_1)\leq z_1,\dots,Z(\textbf{u}_n)\leq z_n|data)$\\	
		&&\\
		&=&$P(Y(\textbf{u}_1)\leq y_1,\dots,Y(\textbf{u}_n)\leq y_n|data)$\\
		&&\\
		&=&$F_{Y(\textbf{u}_1),\dots,Y(\textbf{u}_n)|data}(y_1,\dots,y_n)$\\
		&&\\
		&=&$G_{\boldsymbol{\mu}_{SK}}^{\mathbf{\Sigma}_{SK}}(y_1,\dots,y_n).$\\		
	\end{tabular}
\end{center}

\subsection{Examples}

\subsubsection{The Log-Normal Case}

If $ X $ is a normal random variable with mean $ \mu $ and variance $\sigma^2$, then the random variable $Z=e^X $ is said to be \textit{log-normal} with parameters $\mu$ and $\sigma^2$, or just $Z\sim Logn(\mu,\sigma^2)$. Thus, a random variable $Z$ is log-normal if $ln(Z)$ is normal.

The cumulative distribution function is given by:
$$F_{Z}(z)=G\Bigg(\frac{ln(z)-\mu}{\sigma}\Bigg).$$

It follows that the anamorphosis function is obtained by using \ref{qt}:
$$G\Bigg(\frac{ln(z)-\mu}{\sigma}\Bigg)=G(y),$$
Hence:
\begin{equation}
z=\phi(y)=e^\mu\cdot e^{\sigma y},\quad\quad y=\phi^{-1}(z)=\frac{ln(z)-\mu}{\sigma}.
\end{equation}
From here we obtain:
$$\phi^{'}(y)=\sigma\cdot e^\mu\cdot e^{\sigma y}=\sigma\phi(y),$$
(notice that $\phi^{'}(\phi^{-1}(z))=\sigma\phi(\phi^{-1}(z))=\sigma z$) and that the conditioned local distribution is given by:
\begin{center}
	\begin{tabular}{c c l}
		$f_{Z|data}(z)$&=&$\displaystyle\frac{1}{\phi^{'}(y)}\cdot g_{y^*_{SK}}^{\sigma_{SK}^2}(y)$\\
		&&\\
		&=&$\displaystyle\frac{1}{\sigma\cdot e^{\mu +\sigma y}}\cdot \frac{1}{\smash{\sigma_{SK} \sqrt{2\pi}}}\cdot e^{-\frac{(y-y^*_{SK})^2}{2\sigma_{SK}^2}}$\\ 
		&&\\
		&=&$\displaystyle\frac{1}{z\sigma\sigma_{SK}\sqrt{2\pi}} \cdot e^{-\frac{(ln(z)-\mu-\sigma y^*_{SK})^2}{2\sigma^2\sigma_{SK}^2}}.$\\ 
	\end{tabular}
\end{center}
This is a log-normal distribution with parameters $ \mu+\sigma y^*_{SK} $ and $\sigma^2\sigma_{SK}^2$ for the mean and variance, respectively. Therefore, the local distribution at a certain point $\textbf{u}$ conditioned by the data, which presents a log-normal prior distribution, preserves the log-normality.

Following the same procedure and extending (\ref{e2}) to the $n$-dimesional case, lead us to compute that the local conditional distribution of the vector of random variables $ (Z_1,\dots,Z_n)'$, each one with $Logn(\mu,\sigma^2)$ prior distribution, has a conditional distribution such that $ (X_1,\dots,X_n)' = (ln (Z_1),\dots,ln (Z_n))' $ has an $n$-dimensional
normal distribution with mean vector $\boldsymbol{\mu}= (\mu+\sigma y^*_{SK}(\textbf{u}_1),\dots,\mu+\sigma y^*_{SK}(\textbf{u}_n))' $ and covariance matrix  $\boldsymbol{\Sigma}=(\boldsymbol{\Sigma})_{ij}$, $i,j \in \{1,\dots,n\}$, such that $ (\boldsymbol{\Sigma})_{ij} = (\boldsymbol{\Sigma})_{ji}= \sigma^2\sigma_{SK}(\textbf{u}_i,\textbf{u}_j)$ and $ (\boldsymbol{\Sigma})_{ii} = \sigma^2\sigma_{SK}^2(\textbf{u}_i) $, $i \in \{1,\dots,n\}$, given the data.

In Figure \ref{fig:lognormals}, some possible posterior distributions and their bi-variate behavior are presented.

\begin{figure}[htbp]
	\centering
	\includegraphics[width=0.49\textwidth]{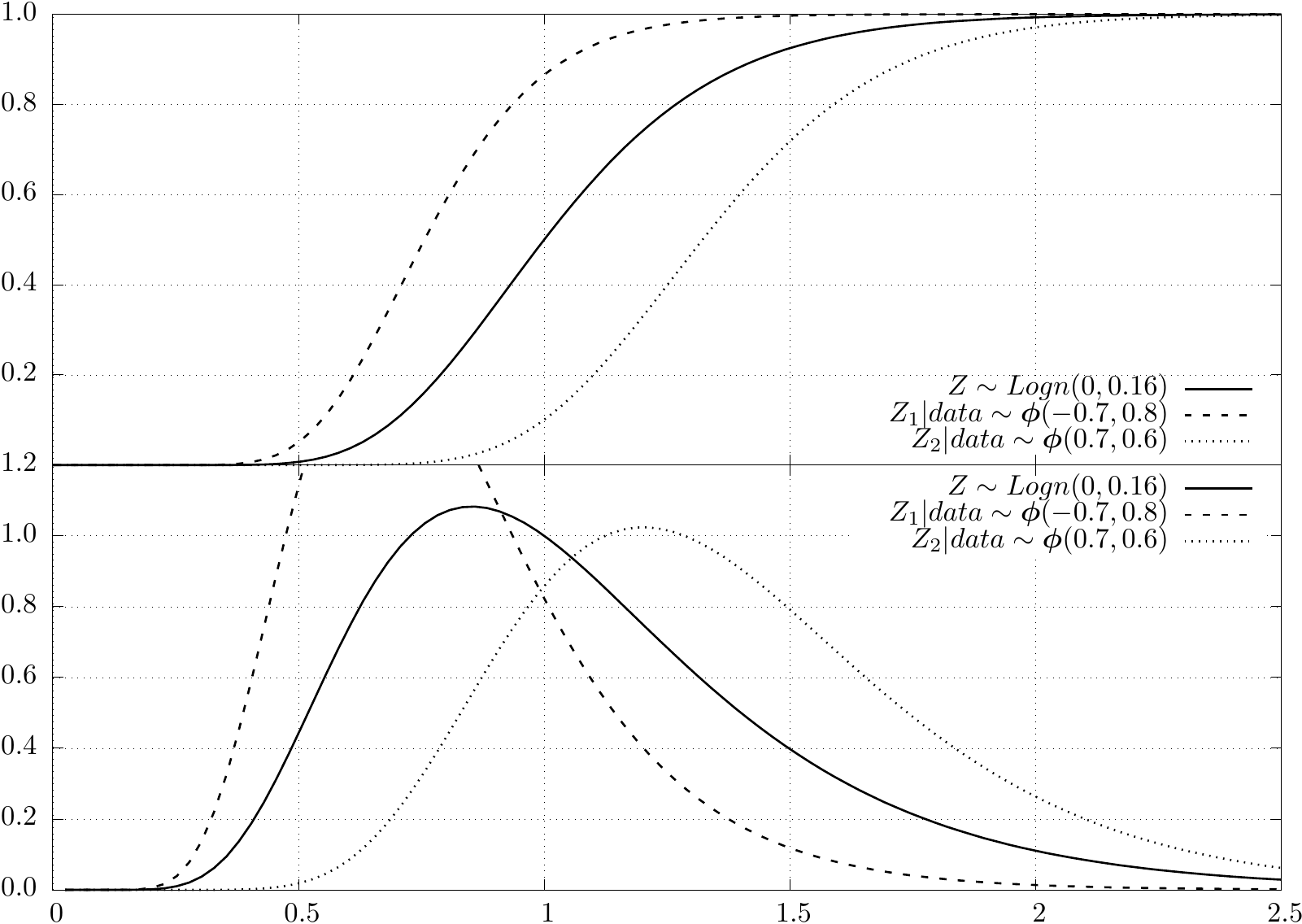}
	\includegraphics[width=0.49\textwidth]{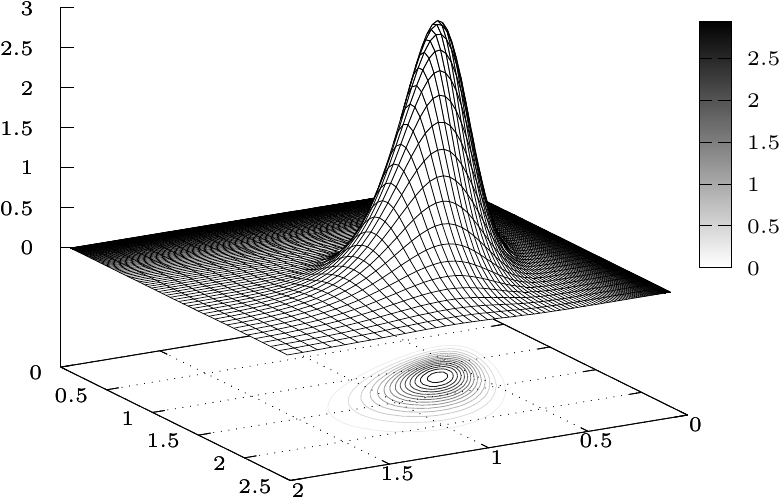}
	
	\caption{Example of probability distributions at a certain pair of points $Z_1$ and $Z_2$ given a log-normal prior distribution, and their bi-variate behavior when the correlation of their Gaussian transformations is $\rho=0.6$.}\label{fig:lognormals}
\end{figure}

\subsubsection{The Exponential Case}

A continuous random variable $ Z $ is said to have an \textit{exponential} distribution with
parameter $ \lambda $, $ \lambda > 0 $, or just $Z\sim Exp(\lambda)$, if its probability density function is given by $f_{Z}(z)=\lambda e^{-\lambda z}$, $z \geq 0$ or, equivalently, if its cdf is given by
$$F_{Z}(z)=1-e^{-\lambda z}, \qquad\qquad  z \geq 0.$$
It follows that the anamorphosis function is:
\begin{equation}
z=\phi(y)=-\frac{1}{\lambda}\cdot ln(1-G(y)).
\end{equation}
From here we obtain:
$$\phi^{'}(y)=\frac{1}{\lambda}\cdot \frac{1}{1-G(y)}\cdot g(y),$$

We will not attempt to find out what kind of probability distribution results of the backtransformation, as we did with the log-normal case. However, numerical results of how some posterior distributions looks like are presented in Figure \ref{fig:exp}. It is very interesting how the exponential distribution is not preserved any more in the posterior distributions (only when the posterior distribution follows a $\mathcal{N}(0,1)$).

\begin{figure}[htbp]
	\centering
	\includegraphics[width=0.49\textwidth]{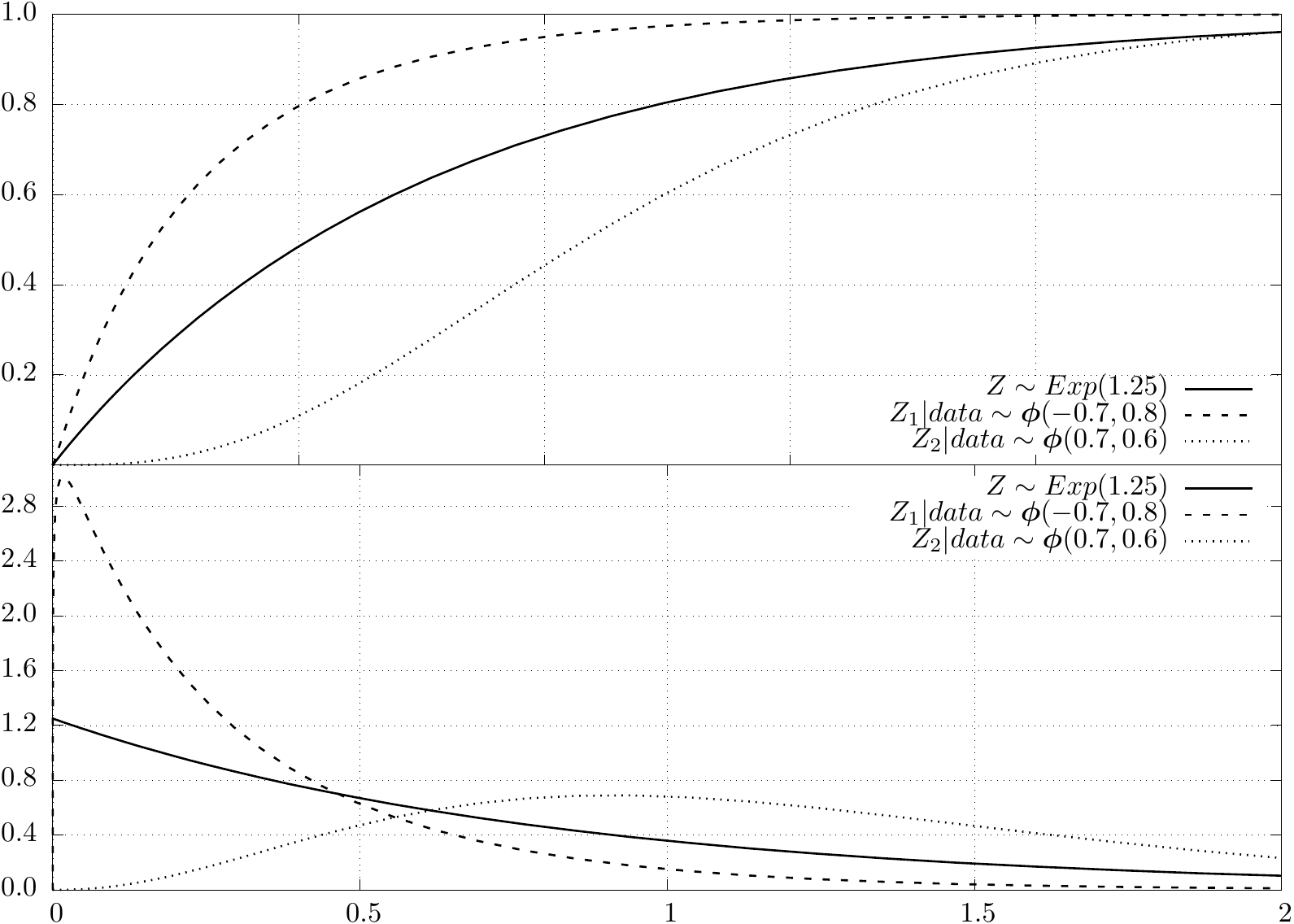}
~ 
	\includegraphics[width=0.49\textwidth]{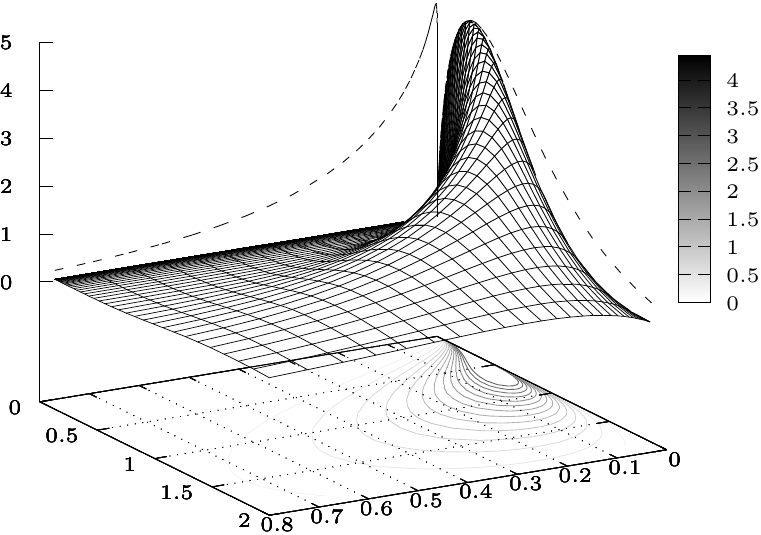}
	
	\caption{Example of probability distributions at a certain pair of points $Z_1$ and $Z_2$ given an exponential prior distribution, and their bi-variate behavior when the correlation of their Gaussian transformations is $\rho=0.6$.}\label{fig:exp}
\end{figure}

We present a synthetic case where the simulation procedure is compared with the results of the direct application of the formulations presented. The synthetic reality for the analysis is generated by LU non-conditional simulation with a known omni directional variogram $\gamma(\mathbf{h})=0.1\cdot Nugg+0.9 \cdot Sph(100)$ in a 100x100 grid of 5x5 units spacing between nodes, for the Gaussian values. Then, the synthetic scenario is randomly sampled, obtaining 100 samples. An exponential distribution with parameter $\lambda=1/0.8$ is used.

\begin{figure}[htbp]
	\centering
	\includegraphics[height=0.32\linewidth]{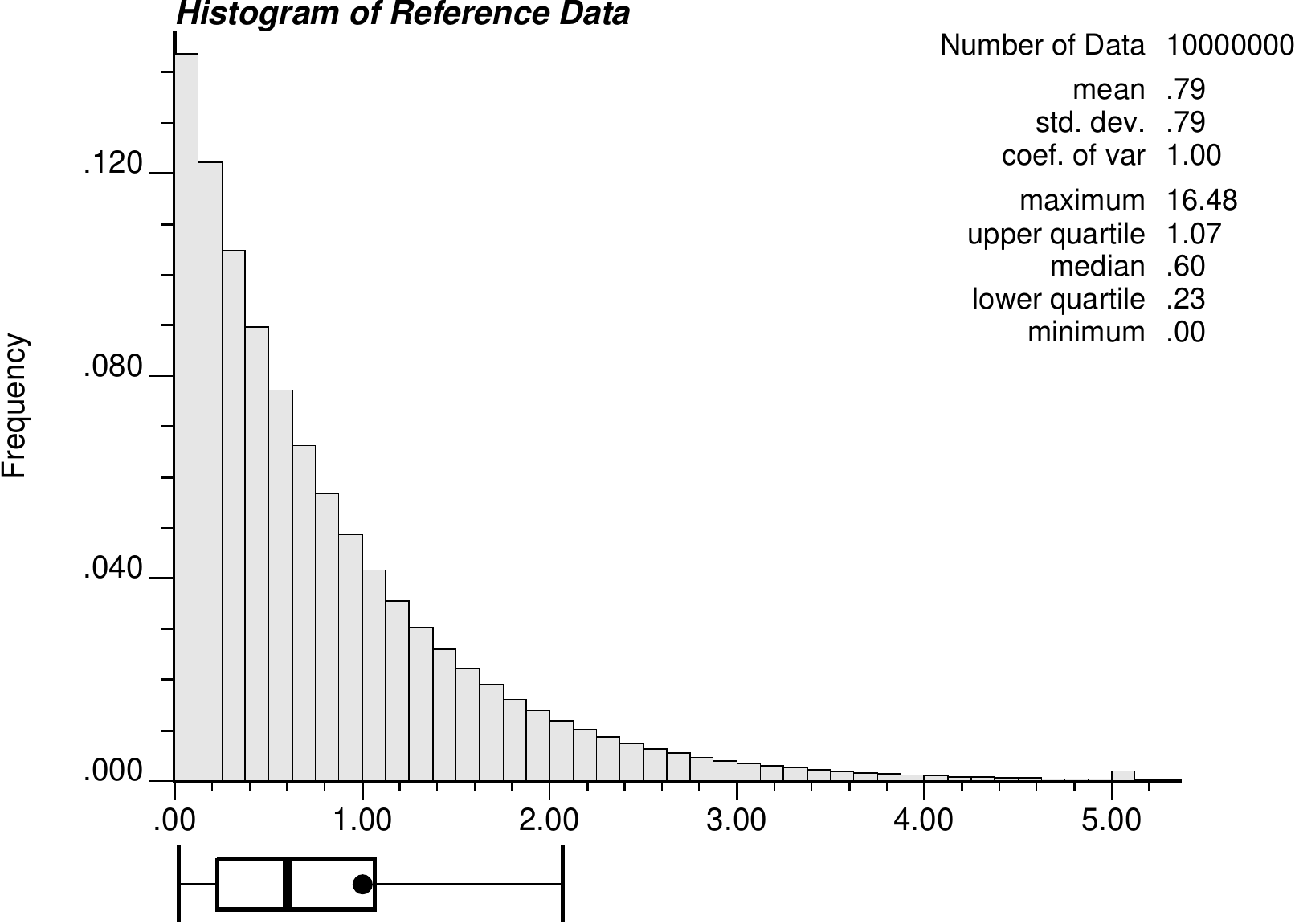}
	\includegraphics[height=0.32\linewidth]{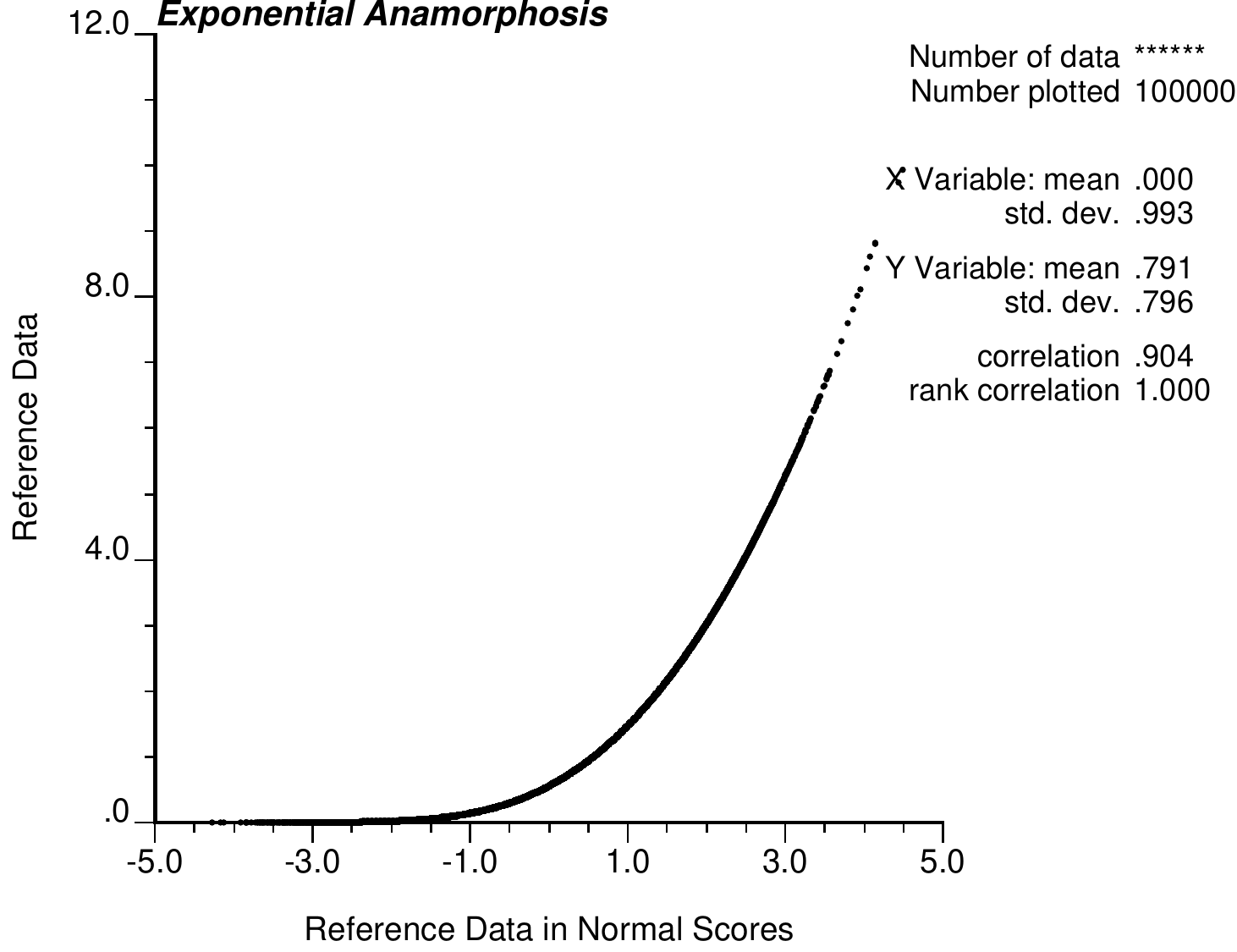}\\
	\includegraphics[width=0.32\linewidth]{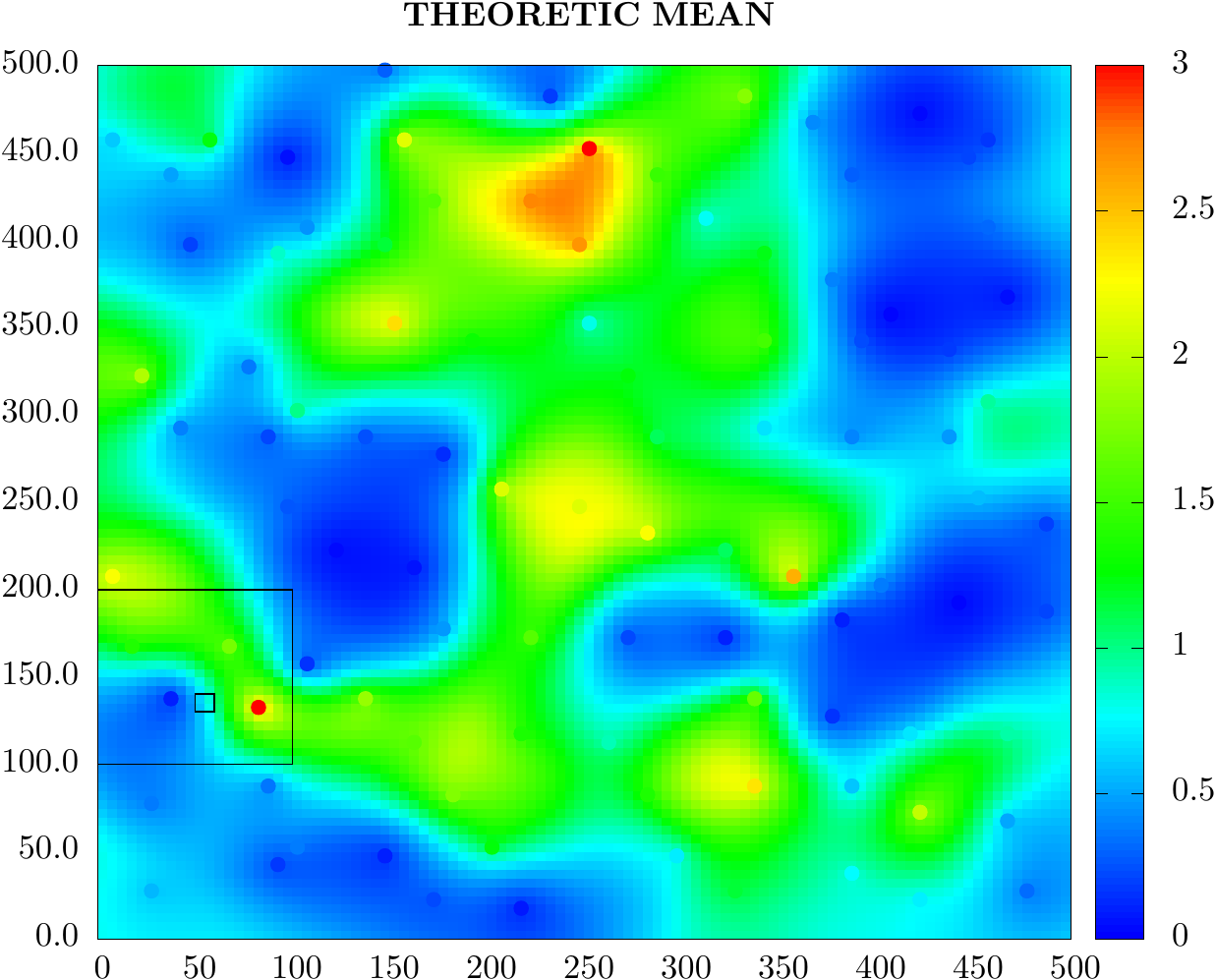}
	\includegraphics[width=0.32\linewidth]{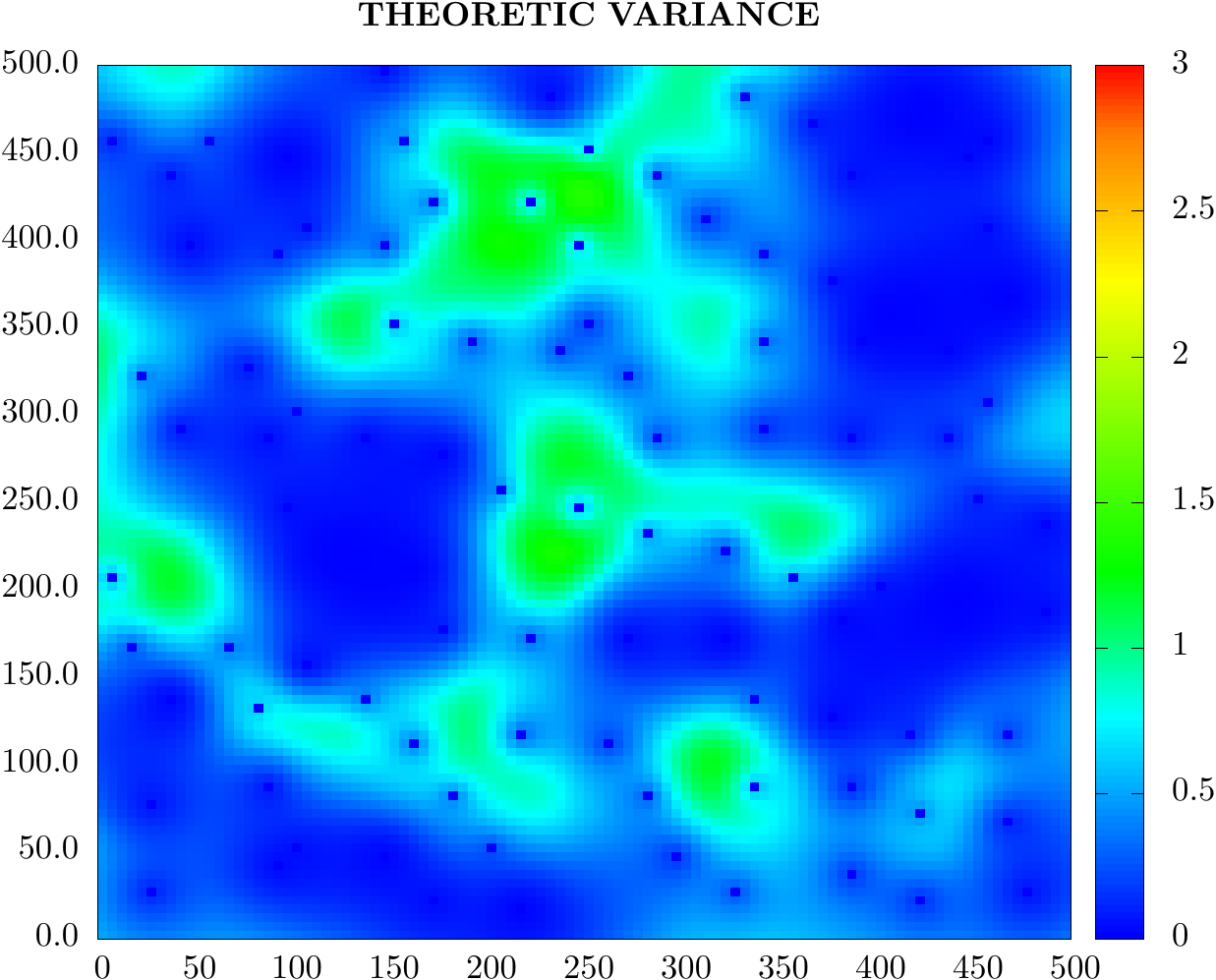} 
	\includegraphics[width=0.32\linewidth]{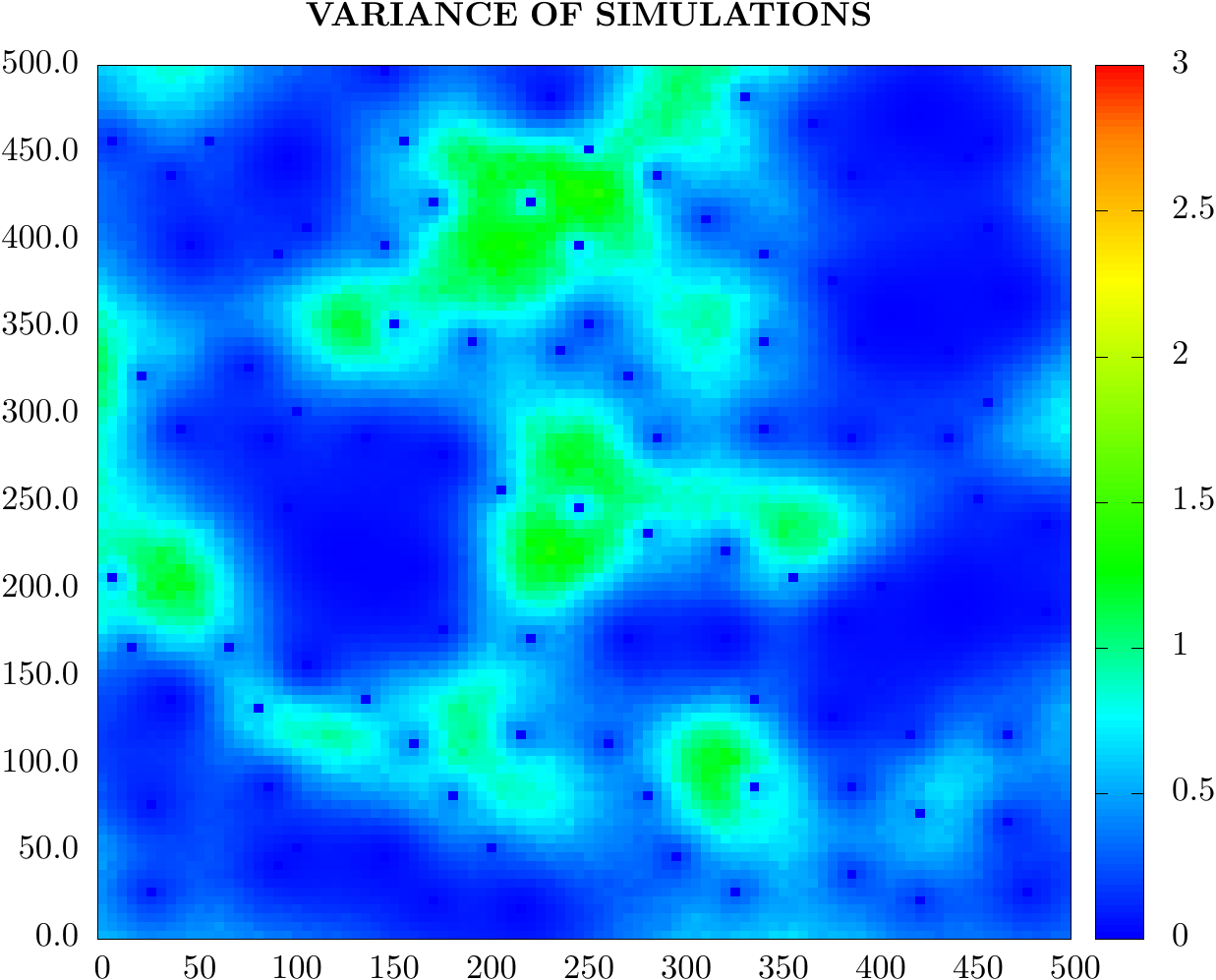}\\
	\includegraphics[width=0.3\linewidth]{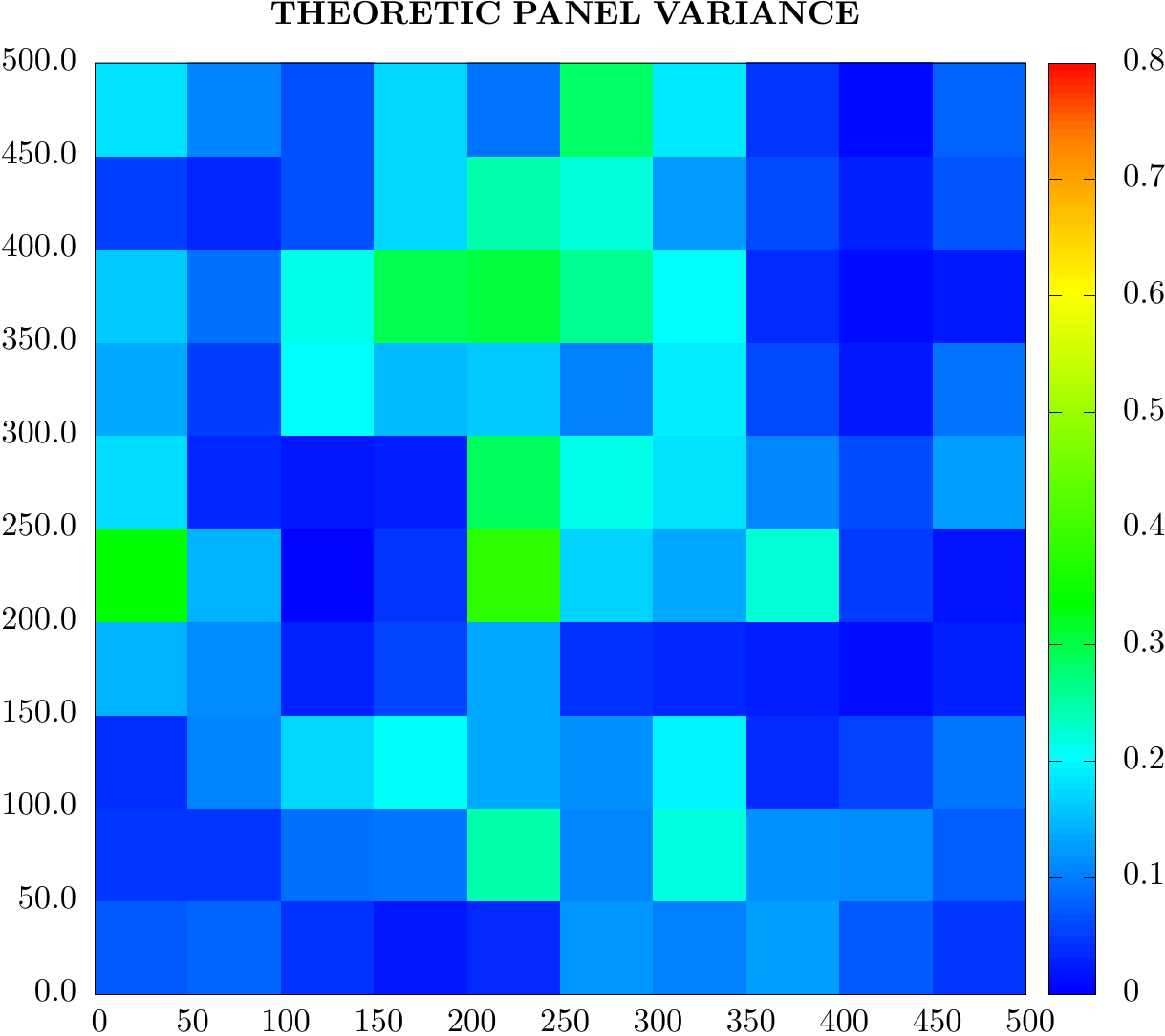} 
	\includegraphics[width=0.3\linewidth]{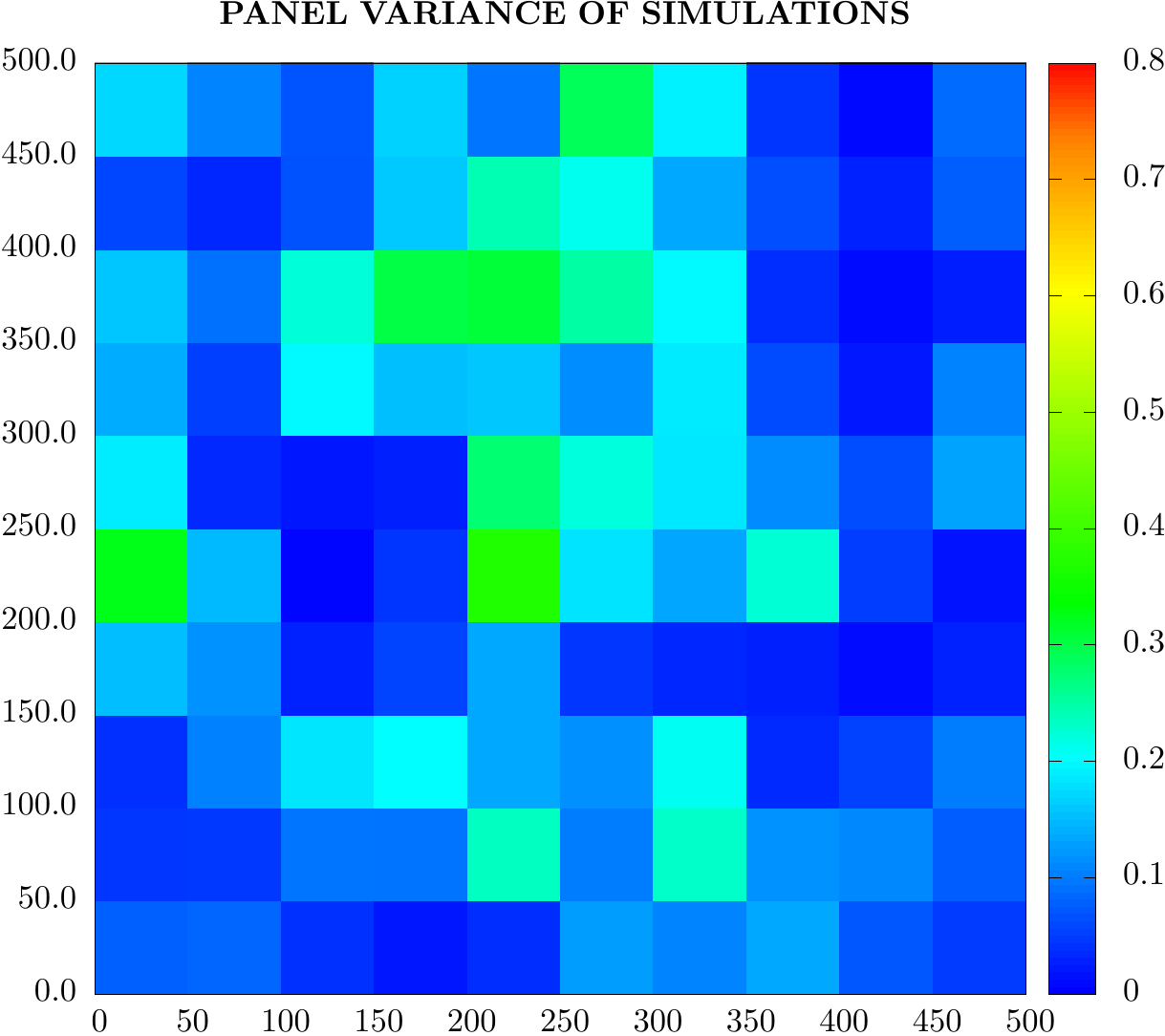}\\
	\includegraphics[width=0.32\linewidth]{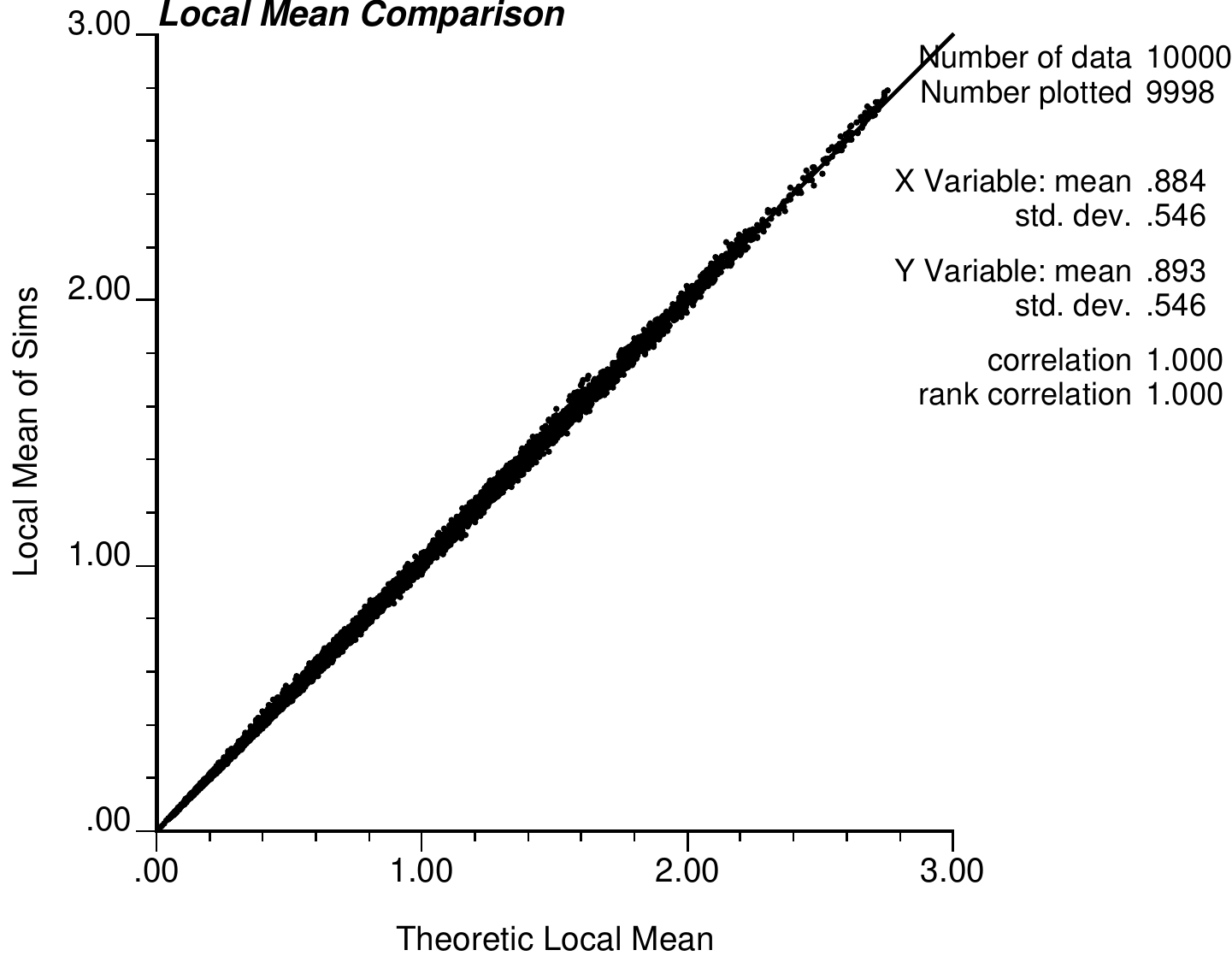}
	\includegraphics[width=0.32\linewidth]{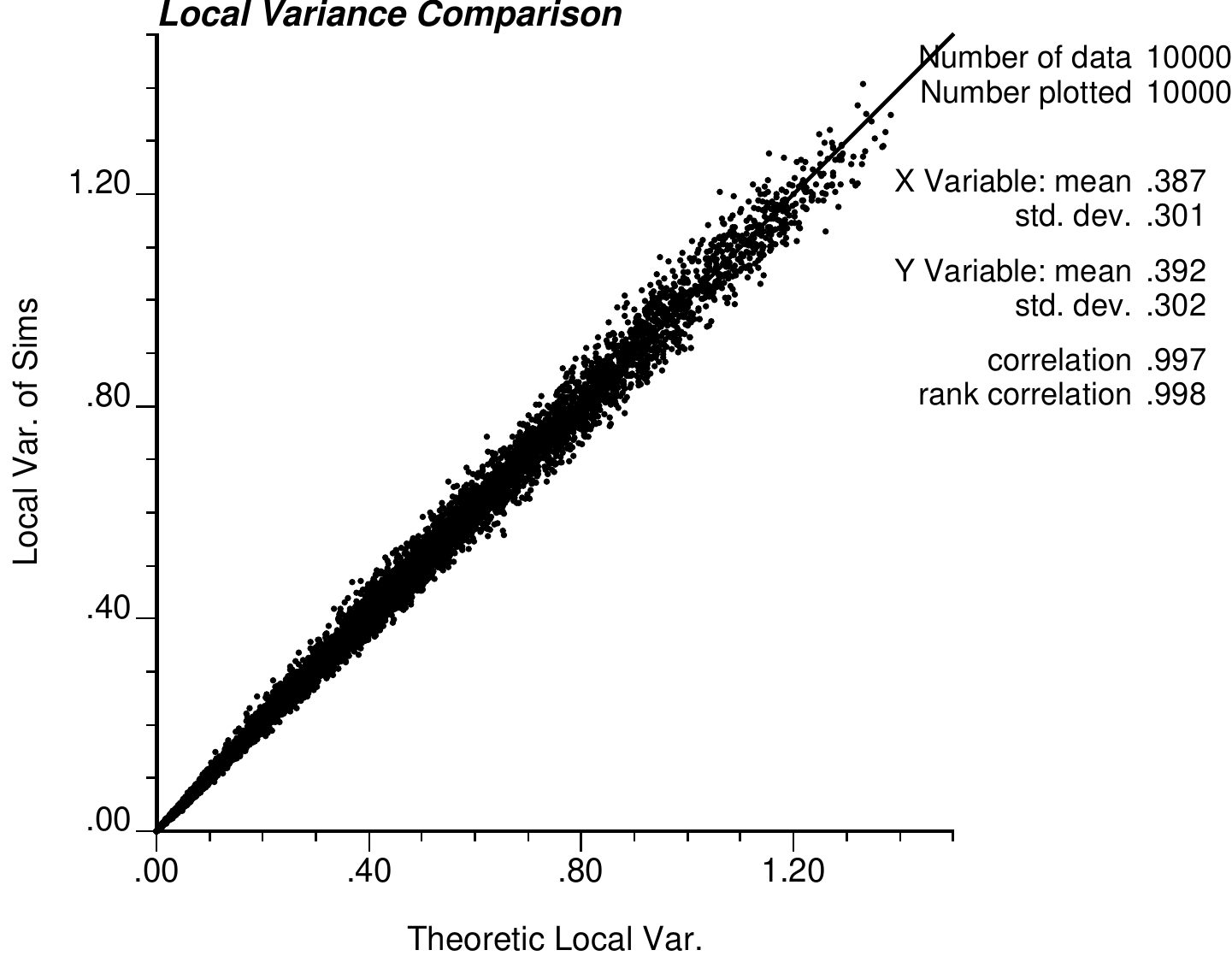}
	\includegraphics[width=0.32\linewidth]{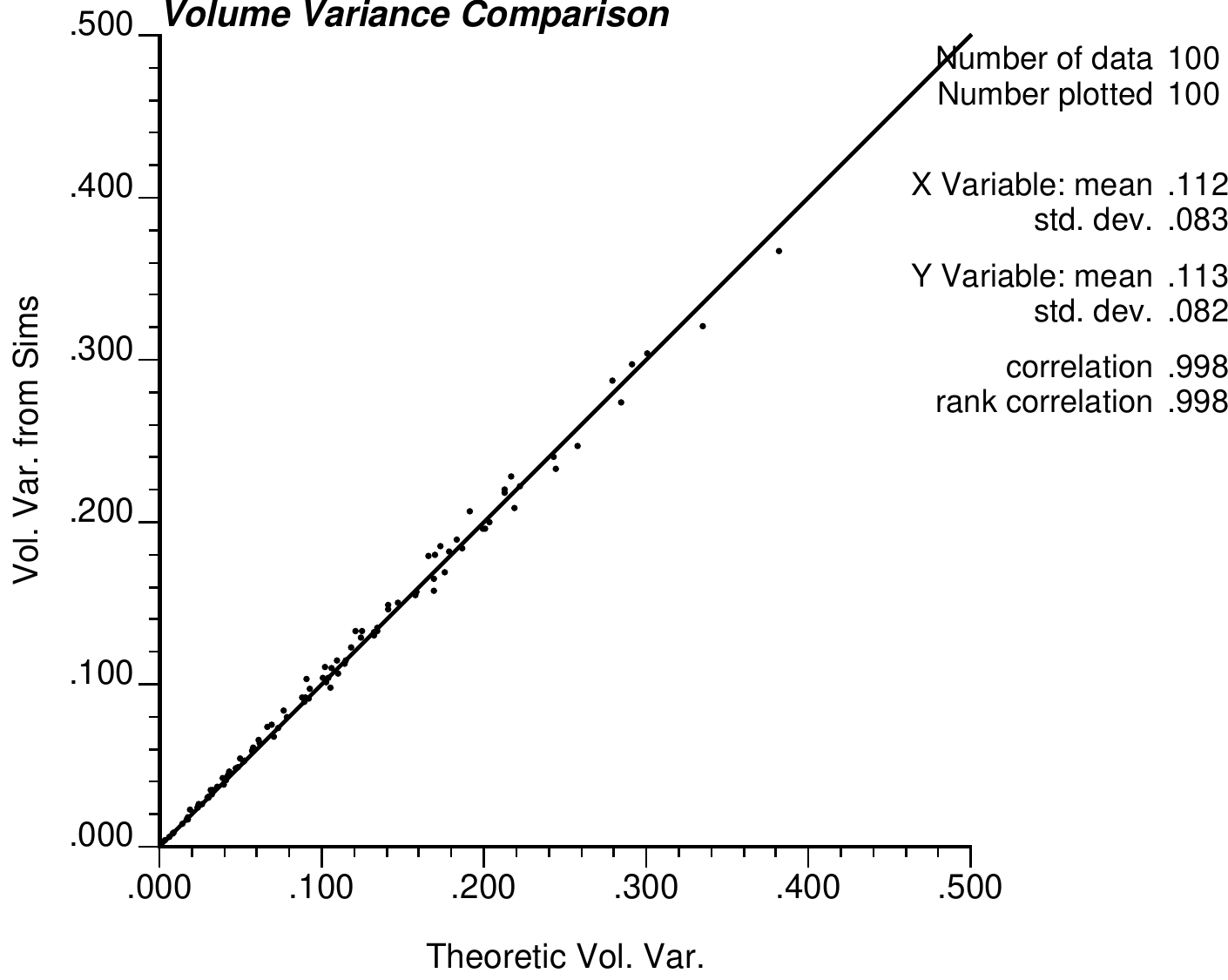}			
	\caption{First row: Reference exponential histogram for the generation of raw values and their anamorphosis. Second row: maps with results of theoretic local mean and variance, and the local variance taking into account 2,000 realizations. The box zone will be zoomed for the example shown in Figure \ref{fig:Res3}. Third row: maps with results of theoretic variance in the volume support (averaging 10x10 points), and the local variance taking into account 2,000 re-blocked realizations. Fourth row: Scatter plots comparing both results theoretic and simulated.}
	\label{img:DataResults}
\end{figure}


\section{A Note Towards the Change of Support}
\label{sec:ChangeSup}

In mining, the actual selection is made on panels, and not on core samples with size $\sim \textbf{u}$. Different panels may be blended into a larger volume $\boldsymbol{V}$ (in downstream processes). Therefore, for the estimation of reserves, it is essential to take the support (size and geometry) of the selection unit into account (resulting in several hundred or thousand tones). Otherwise there is a risk of bias with serious economic consequences.

In general, this support will be considerably different from the support of the exploration unit (core samples for example). This is one of the most important aspects on the economic and technical context of a mining project, with consequences in the amount of recoverable reserves. The entire resources of a deposit are rarely mineable because of the variability of mining and treatment costs added to the spatial variability of ore quality. For a detailed explanation of the issue, the reader can consult the survey by Huijbregts \cite{Huijbregts}. In a nutshell, consider the histogram of dispersion of the grades $Z(\textbf{u})$ as on Figure \ref{img:CS}.
The histogram of grades of all the mining units of larger support will have: (i) a mean equal to the mean of the core sample grades; (ii) an experimental dispersion variance less than the core support; and (iii) for high cutoff grades above the mean, $z>m$, the hatched area of the histogram of core grades may seriously overestimate the possible mined recovery in both tonnage and mean grade and, correspondingly, underestimate the tonnage of metal left in the waste, i.e., underestimate the area corresponding to the panels with true grades $z<m$.

\begin{figure}[htbp]
	\centering
	\includegraphics[height=0.5\linewidth]{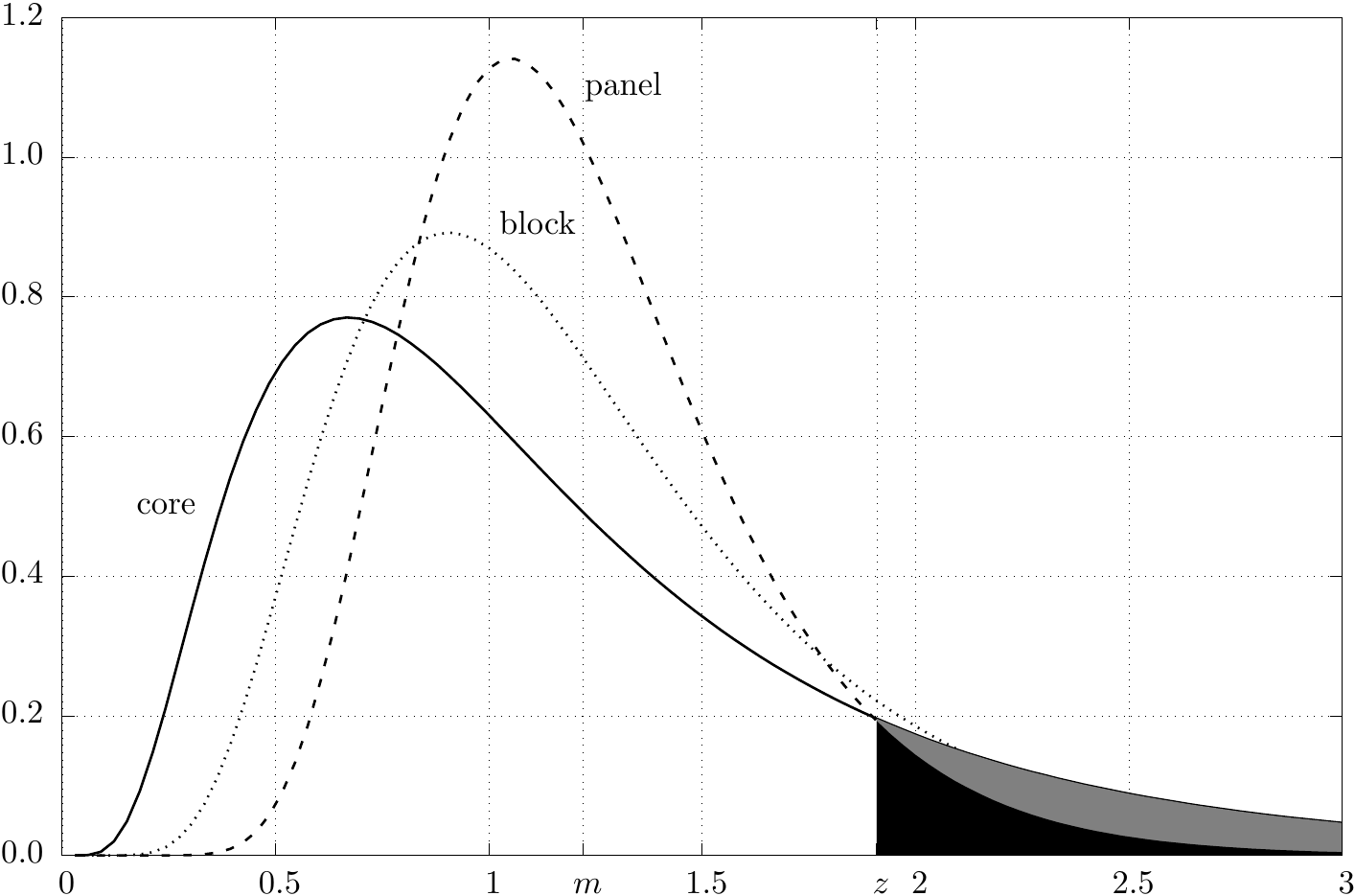}		
	\caption{Change in the variable probability distribution when different supports are considered.}
	\label{img:CS}
\end{figure}

In what follows, we present an extension of the multi-Gaussian model to address the distribution when a change of support is considered, locally. We are interested in finding the probability distribution of the main variable within a given volume of the spatial domain. For this, let's say the volume is approximated by a discretization of $n$  points, located at $\textbf{u}_i$ positions, $i=1,\dots,n$, each one with mass density 
$\rho(\textbf{u}_i)$, $\sum_{i=1}^{n}\rho(\textbf{u}_i)=1$, we are interested in finding the probability distribution of the random variable:
\begin{equation}\label{npoints}
Z^*_{\boldsymbol{V}}= \rho(\textbf{u}_1)Z^*(\textbf{u}_1) + \rho(\textbf{u}_1)Z^*(\textbf{u}_2)+ \dots + \rho(\textbf{u}_n)Z^*(\textbf{u}_n),
\end{equation}
recalling that the point-support random variables are correlated, and they have to be conditioned by the sampled data, which, present themselves in different configurations and with different conditioning values, leading to different point-support distributions.

In this discretized framework, we can consider $\rho(\textbf{u})$ constant within the geological volume, and we omit, without loss of generality, this constant for the next steps. The problem can be addressed as follows.

Let's begin by expressing the cumulative  distribution of the sum:
\begin{center}
	\begin{tabular}{c c l}
		$F_{Z^*_{\boldsymbol{V}}}(z)$&=&$P(Z^*_{\boldsymbol{V}}\leq z)$\\
		&&\\
		&=&$P(Z_{\boldsymbol{V}}\leq z|data)$\\ 
		&&\\
		&=&$P(Z(\textbf{u}_1)+ \dots + Z(\textbf{u}_n)\leq z|data)$\\ 
		&&\\
		&=&$\displaystyle\idotsint\displaylimits_{z_1+\dots+z_n\leq z} f_{Z_1,\dots,Z_n|data} (\mathbf{z_1},\dots,\mathbf{z_n}) d\mathbf{z_1}\dots d\mathbf{z_n}$
	\end{tabular}
\end{center}
This can be rewritten in term of the $Y$ variable and the anamorphosis function, for a given value $a$ such that $z=\phi(a)$:
\begin{center}
	\begin{tabular}{c c l l l}
		$F_{Z^*_{\boldsymbol{V}}}(z)$&=&$\displaystyle\idotsint\displaylimits_{z_1+\dots+z_n\leq z} f_{Z_1,\dots,Z_n|data} (\mathbf{z_1},\dots,\mathbf{z_n}) d\mathbf{z_1}\dots d\mathbf{z_n}$&&\\
		&&&&\\
		&=&$\displaystyle\idotsint\displaylimits_{\phi(y_1)+\dots+\phi(y_n)\leq \phi(a)} g_{\boldsymbol{\mu}_{SK}}^{\mathbf{\Sigma}_{SK}} (\mathbf{y_1},\dots,\mathbf{y_{n-1}},\mathbf{y_n}) d\mathbf{y_n}d\mathbf{y_{n-1}}\dots d\mathbf{y_1}$&=&$F_{Z^*_{\boldsymbol{V}}}(\phi(a)).$\\ 
	\end{tabular}
\end{center}
Then, the probability density can be determined joining the following two developments. On one hand:

\begin{align*}
	&\displaystyle\frac{d}{da}F_{Z^*_{\boldsymbol{V}}}(\phi(a))=f_{Z^*_{\boldsymbol{V}}}(\phi(a))\cdot\phi'(a),
\end{align*}
and on the other hand:

\begin{align*}
	&\displaystyle\frac{d}{da}F_{Z^*_{\boldsymbol{V}}}(\phi(a))\\
	&=\displaystyle\frac{d}{da}\int_{-\infty}^{\infty}\dots\int_{-\infty}^{\infty}\int_{-\infty}^{\phi^{-1}\{\phi(a)-\phi(y_1)+\dots+\phi(y_{n-1})\}} g_{\boldsymbol{\mu}_{SK}}^{\mathbf{\Sigma}_{SK}} (\mathbf{y_1},\dots,\mathbf{y_{n-1}},\mathbf{y_n})d\mathbf{y_n}d\mathbf{y_{n-1}}\dots d\mathbf{y_1}\\
	&=\displaystyle\int_{-\infty}^{\infty}\dots\int_{-\infty}^{\infty}\frac{d}{da}\int_{-\infty}^{\phi^{-1}\{\phi(a)-\phi(y_1)+\dots+\phi(y_{n-1})\}} g_{\boldsymbol{\mu}_{SK}}^{\mathbf{\Sigma}_{SK}} (\mathbf{y_1},\dots,\mathbf{y_{n-1}},\mathbf{y_n}) d\mathbf{y_n}d\mathbf{y_{n-1}}\dots d\mathbf{y_1}\\
	&=\displaystyle\int_{-\infty}^{\infty}\dots\int_{-\infty}^{\infty}\frac{g_{\boldsymbol{\mu}_{SK}}^{\mathbf{\Sigma}_{SK}}(\mathbf{y_1},\dots,\mathbf{y_{n-1}},\phi^{-1}\{\phi(a)-\phi(\mathbf{y_1})+\dots+\phi(\mathbf{y_{n-1})}\})}{\phi'[\phi^{-1}\{\phi(a)-\phi(\mathbf{y_1})+\dots+\phi(\mathbf{y_{n-1}})\}]}\cdot \phi'(a)d\mathbf{y_{n-1}}\dots d\mathbf{y_1}.\\
\end{align*}

This shows that we are moving from the volume of integration $\mathcal{V}(a)$ parameterized by $\mathcal{V}(a)=\{z_1,\dots,z_n|z_1+\dots+z_n\leq \phi(a)\}$, when we compute $F_{Z^*_{\boldsymbol{V}}}$, to the surface of this volume $\mathcal{S}(a)=\partial\mathcal{V}(a)$ when we  compute $f_{Z^*_{\boldsymbol{V}}}$, which defines planes in term of the $Z$ variable but not necessarily in terms of the Gaussian variable $Y$.

Notice that $\mathbf{y_n}$ is now determined by the rest of the variables. Thus, we can just rewrite $\phi^{-1}\{\phi(a)-\phi(\mathbf{y_1})+\dots+\phi(\mathbf{y_{n-1}})\}=\mathbf{y_n}(a,\mathbf{y_{1}},\dots,\mathbf{y_{n-1}})={y_n}$. Therefore,
\begin{equation}\label{reblock}
f_{Z^*_{\boldsymbol{V}}}(\phi(a))=\int_{\mathcal{S}(a)}\frac{g_{\boldsymbol{\mu}_{SK}}^{\mathbf{\Sigma}_{SK}}(\mathbf{y_1},\dots,\mathbf{y_{n-1}},{y_n})}{\phi'({y_n})}d\mathbf{y_{n-1}}\dots d\mathbf{y_1}.
\end{equation}

The procedure described is summarized in Figure \ref{fig:Trans}, where we want to compute the sum of two correlated lognormal random variables. The straight lines $z_2=z-z_1$ where integration needs to be done, in the bi-lognormal distribution space, are ``bent'', together with the density itself, being transformed into curved lines $\phi(y_2)=\phi(a)-\phi(y_1)$ in the Gaussian framework. 
\begin{figure}[htbp]
	\centering
	\includegraphics[height=0.43\textwidth,trim={0 0 1.5cm 0},clip]{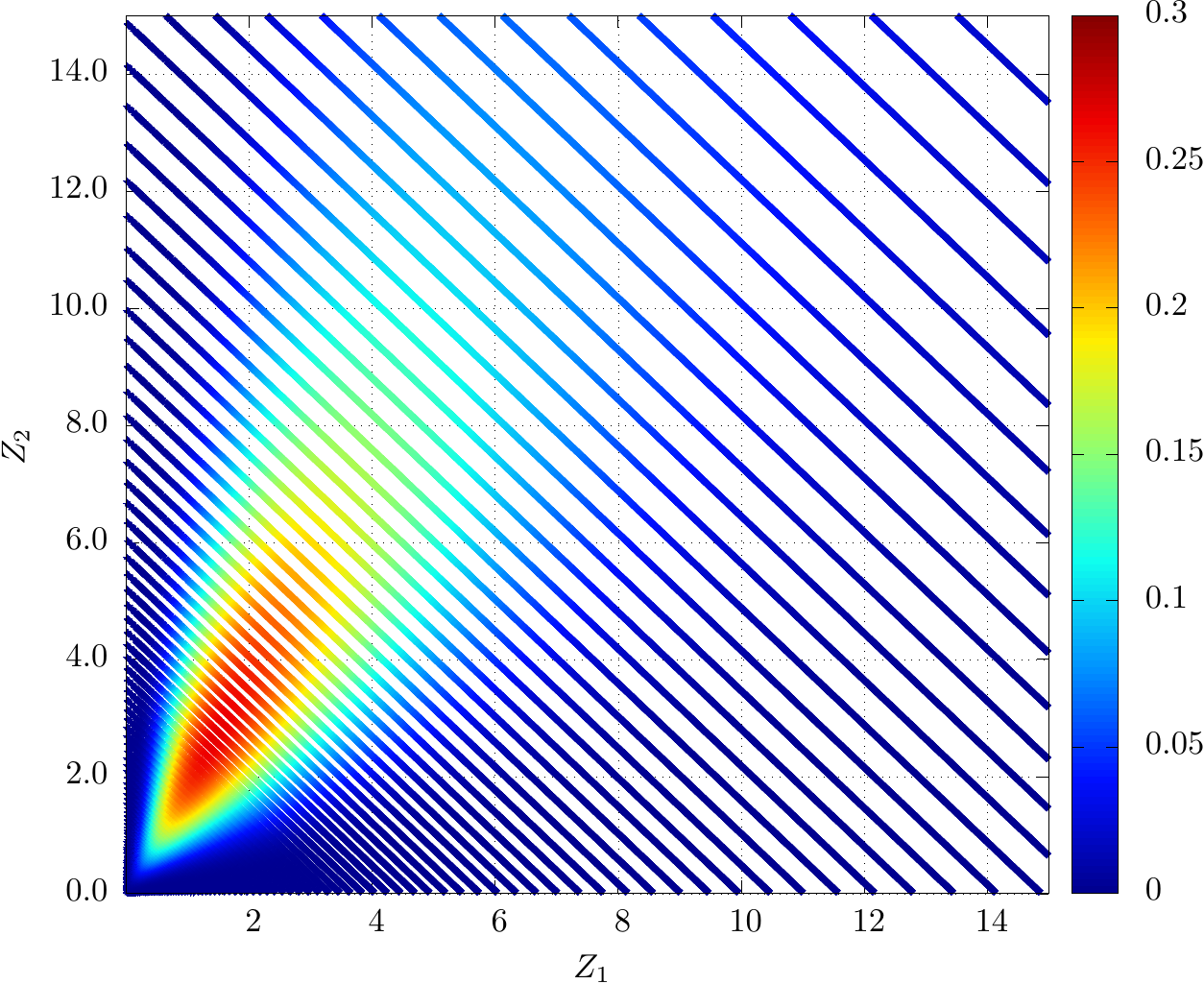}
	\includegraphics[height=0.43\textwidth]{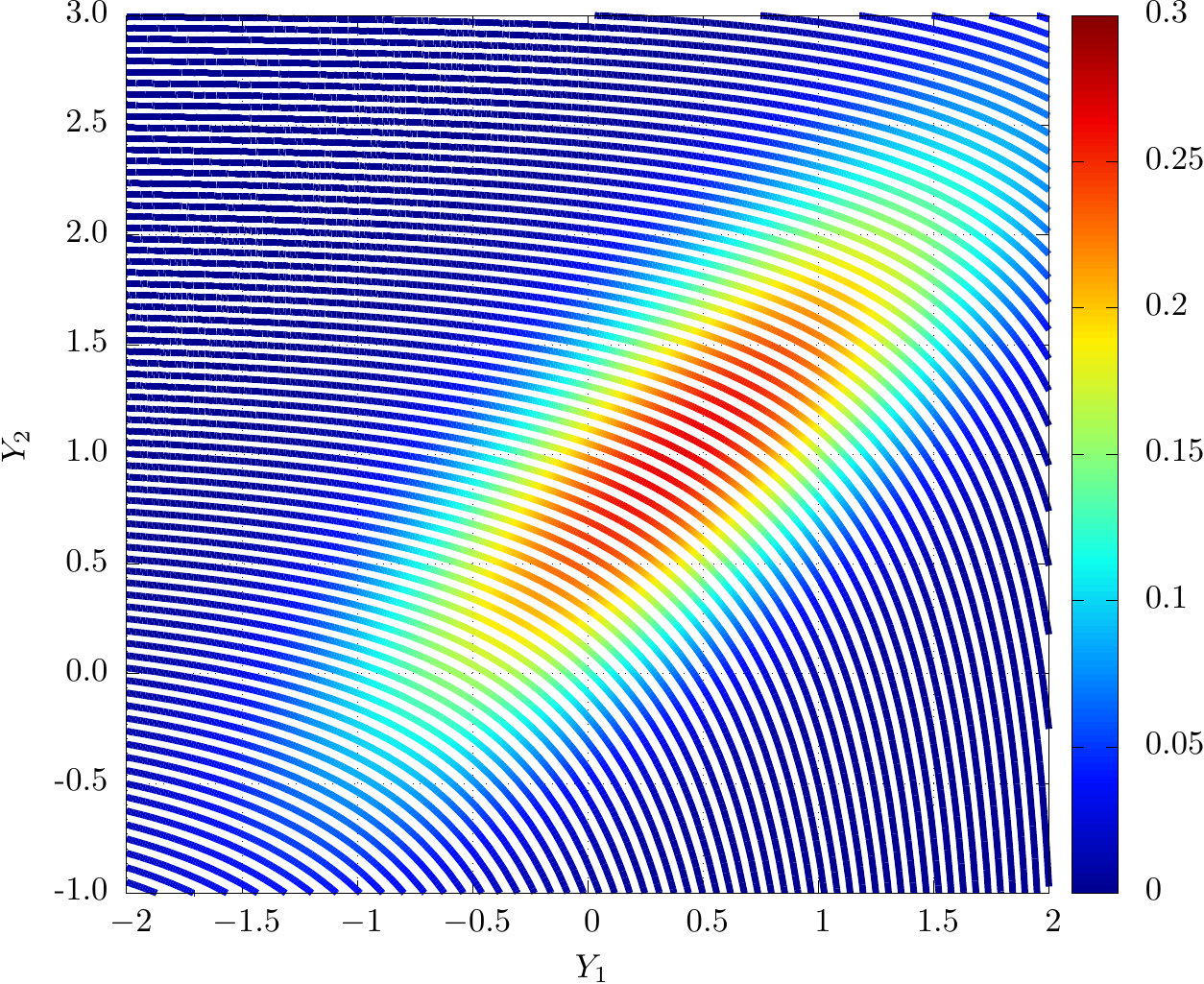}
	\caption{Graphical representation for computing the exact probability density of the sum of two correlated log-normal variables, pointing towards the change of support. Each line represent the domain of integration for certain value of $a$, along the Gaussian density should be computed and integrated to obtain a certain value $f_{Z^*_{\boldsymbol{V}}}(\phi(a))$.}\label{fig:Trans}
\end{figure}

\begin{Rk}
Although this procedure gives exact results, it is unfeasible in practice because of its computational complexity, requiring to integrate the multi-Gaussian density in $\mathcal{S}(a)\subset\R^{n-1}$. However, a highly suitable alternative could be to consider Monte-Carlo integration of the integral \ref{reblock}. One can argue that this would be similar to using geostatistical simulations. It is not, since we can obtain immediately answers such as the probability above a cut-off value $\textit{\textbf{z}}$, $P(Z^*_{\boldsymbol{V}}>\textit{\textbf{z}})$, by computing $f_{Z^*_{\boldsymbol{V}}}(z)$ in a discrete set of points in the interval $(0,\textit{\textbf{z}})$, in contrast with geostatistical simulations, where there is no control in the outcomes range and values.
\end{Rk}

We show an example of reblocking, where the exact distribution of the points and the average of points are presented using expression \ref{reblock}, together with the outcomes of the simulation, at point and block support. We have been able to implement it with $n=4$, in the exponential example of Figure \ref{img:DataResults}. The results are summarized in Figure \ref{fig:Res3}.

\begin{figure}[htbp]
	\centering 
	\includegraphics[align=c,width=0.53\textwidth]{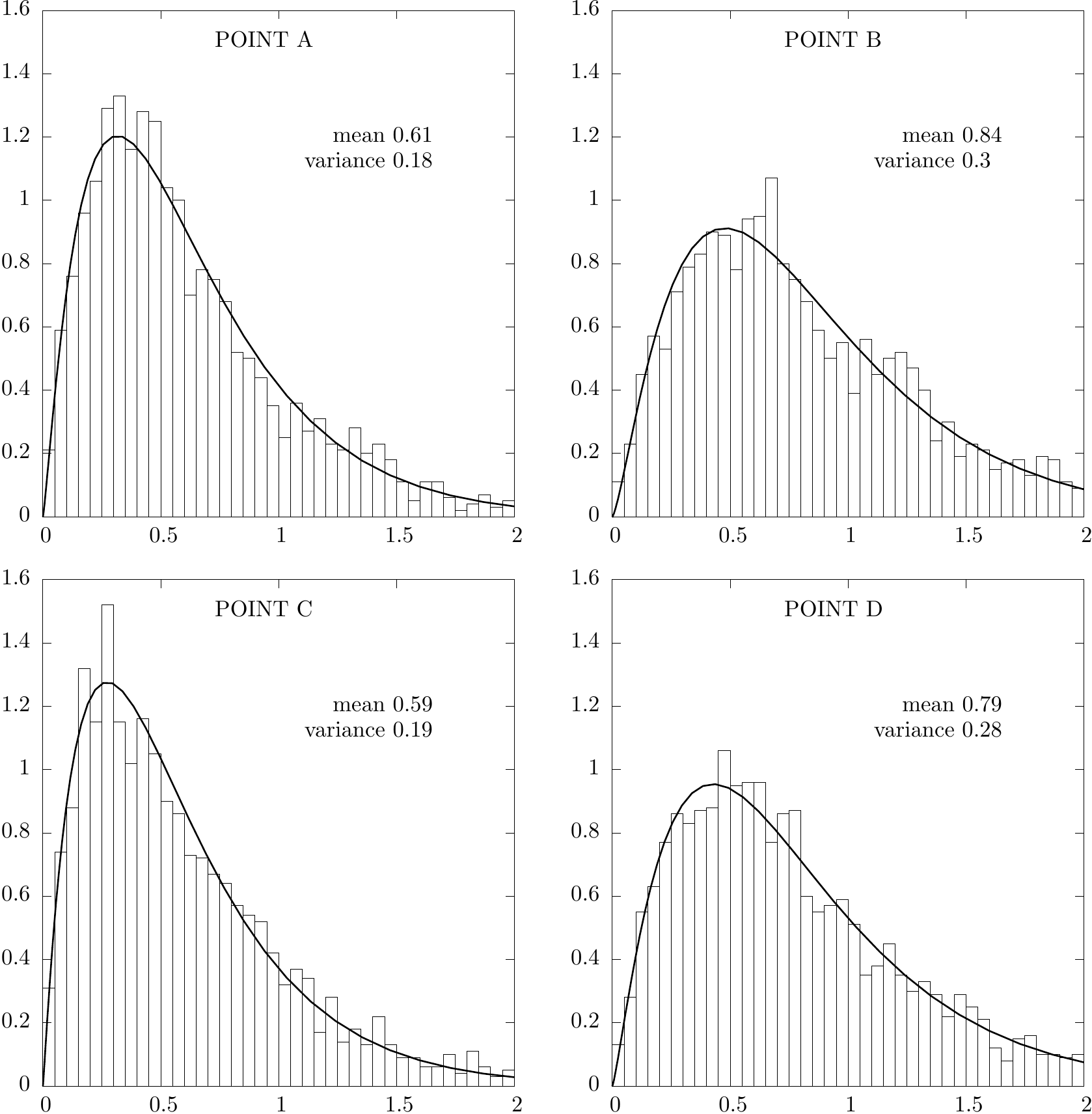}
	\includegraphics[align=c,width=0.46\textwidth]{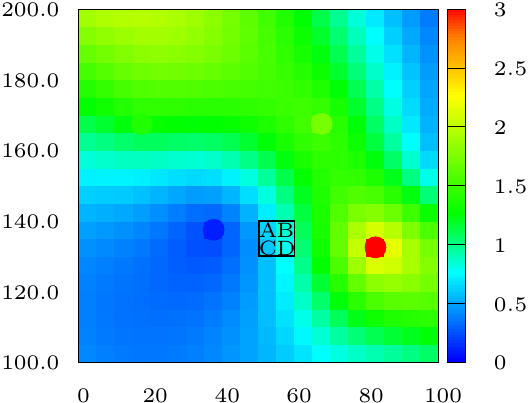}
	\raggedleft
	\includegraphics[align=c,width=0.255\textwidth]{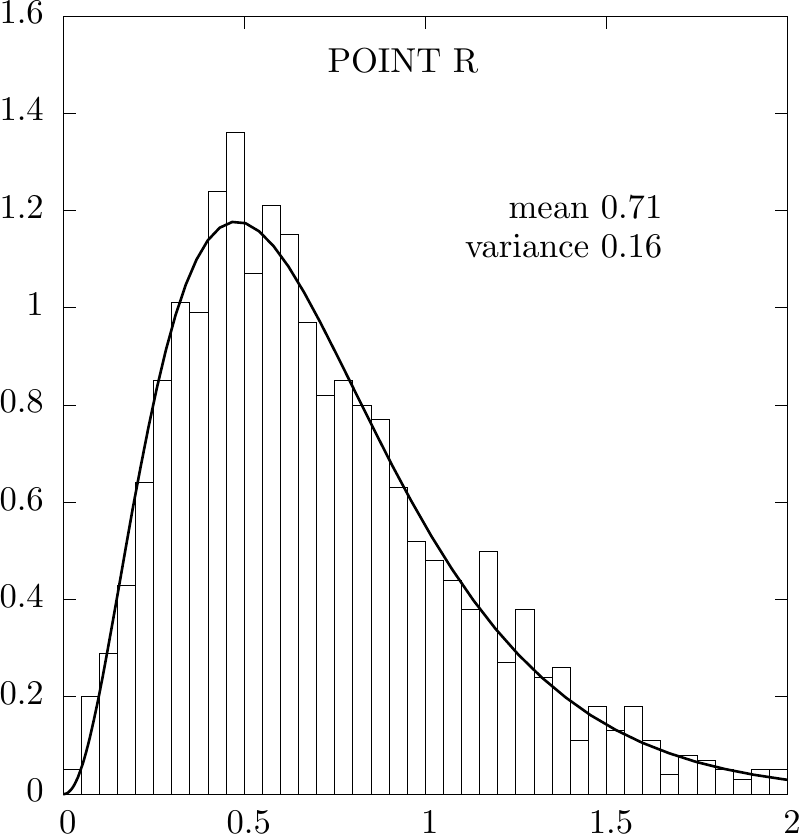}
	\includegraphics[align=c,width=0.46\textwidth]{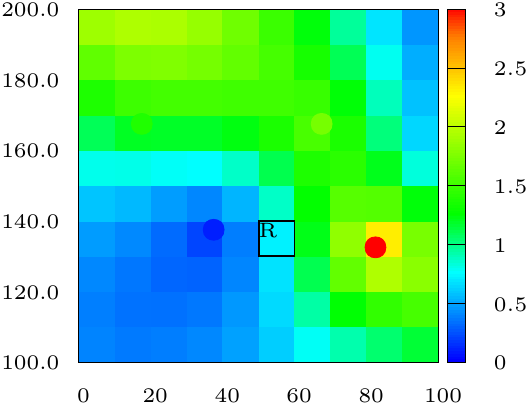}
	\caption{Top: Histogram of 2,000 value scenarios at points A, B, C and D, together with the analytic backtransformed distribution according to eq. \ref{back}, for the exponential case. Bottom:  Histogram of the mean value at points A, B, C and D, in each of the 2,000 scenarios, together with the analytic backtransformed distribution according to eq. \ref{reblock}.}
	\label{fig:Res3}
\end{figure}

\begin{Rk}
If the spatial discretization is infinitesimal, expression \ref{npoints} can be formally expressed as:
\begin{align}\label{extreme}
Z^*_{\boldsymbol{V}}=\int_{\boldsymbol{V}}\rho(\textbf{u})Z^*(\textbf{u})d\textbf{u},
\end{align}
with $\int_{\boldsymbol{V}}\rho(\textbf{u})d\textbf{u}=1$. Supposing we are considering an anamorphosis transformation given by an initial log-normal distribution with mean $0$ and variance $1$, then \ref{extreme} can be expressed as:
$$Z^*_{\boldsymbol{V}}=\int_{\boldsymbol{V}}\rho(\textbf{u})e^{Y^*(\textbf{u})}d\textbf{u},$$
where $Y^*(\textbf{u})$ is the conditioned Gaussian random variable. Some approximations to this form of the problem can be found in \cite{matheron1974effet}.
\end{Rk}

\section{Hermite Polynomials}
\label{sec:Her}

In this section Hermite Polynomials are introduced as an infinite sum alternative (truncated at some term $P$), suitable for avoiding numerical integration in some of the formulations already presented. Firstly, we begin by introducing Hermite Polynomials, and secondly, we present a self-contained list of proofs for most of the important results that we will use in the following section. The proofs are taken from \cite{Lantuejoul1,matheron1974fonctions,andrews2000special,Emery2005a}.

Hermite Polynomials $H_n(Y(\textbf{u}))$ are polynomials that have special properties related to the normal distribution. Hermite Polynomials are defined by Rodrigues' formula:

\begin{equation}
\label{hermite}
H_n(y)=\frac{1}{\sqrt{n!}} \cdot \frac{g^{(n)}(y)}{ g(y)}\qquad \forall n\ge 0
\end{equation}

where $g^{(n)}(y)=\frac{d^n g(y)}{d y^n}$, $\displaystyle\frac{1}{\sqrt{n!}} $ is a normalization factor, $ y $ is a Gaussian
value, and $ g(y) $ is the standard Gaussian probability distribution function (defined in Equation \ref{Gauss}). The Hermite Polynomial $H_n(Y(\textbf{u}))$ is a polynomial of degree $n$. More specifically:

\begin{center}
	\begin{tabular}{c c l}
		$ H_0(y) $&=&$ 1 $\\
		&&\\
		$ H_1(y) $&=&$ -y $\\
		&&\\
		$ H_2(y) $&=&$ \displaystyle \frac{1}{\sqrt{2}}\cdot(y^2-1) $\\
	\end{tabular}
\end{center}

\begin{Thm}
	$ H_n $ is a polynomial of degree $n$ that satisfy the recurrence relations:
	$$\sqrt{(n+1)}\cdot H_{n+1}(y)+yH_{n}(y)+\sqrt{n}\cdot H_{n-1}(y)=0$$
	$$H'_n(y)=-\sqrt{n}\cdot H_{n-1}(y)$$
\end{Thm}
\begin{proof}
	Calculating the successive derivatives of the Gaussian density:
	$$g'(y)+yg(y)=0$$
	$$g''(y)+yg'(y)+g(y)=0$$
	$$g^{(3)}(y)+yg''(y)+2g(y)=0$$	
	and, in general
	\begin{equation}
	\label{hermite1}
	g^{(n+1)}(y)+yg^{(n)}(y)+ng^{(n-1)}=0
	\end{equation}
	Noticing that definition (\ref{hermite}) implies:
	$$g^{(n)}=\sqrt{n!}\cdot g(y)\cdot H_{n}(y),$$
	and plugging this relation in (\ref{hermite1}) (and then dividing by $g(y)$) leads to:
	\begin{equation}
	\label{hermite2}
	\sqrt{(n+1)!} H_{n+1}(y)+y\sqrt{n!}H_{n}(y)+\sqrt{n}\sqrt{n!} H_{n-1}(y)=0.
	\end{equation}
	
	Dividing by $\sqrt{n!}$ leads to the first relation. From here, we obtain a way to generate the Hermite polynomials:
	\begin{equation}
	H_{n+1}(y)=-\frac{1}{\sqrt{n+1}}yH_{n}(y)-\sqrt{\frac{n}{n+1}}H_{n-1}(y) \nonumber
	\end{equation}
	
	On the other hand:
	\begin{align}
	\sqrt{(n+1)!} H_{n+1}(y) g(y)=g^{(n+1)}(y)&=(g^{(n)}(y))'\nonumber\\
	&=(\sqrt{n!} H_{n}(y) g(y))'=\sqrt{n!} H'_{n}(y) g(y)-y\sqrt{n!} H_{n}(y) g(y) \nonumber
	\end{align}
	which gives us a second recurrence relation:
	\begin{align}\label{hermite3}
	\sqrt{(n+1)!} H_{n+1}(y)+y\sqrt{n!} H_{n}(y)-\sqrt{n!} H'_{n}(y)=0
\end{align}
	 Comparing between these two recurrence relations, (\ref{hermite2}) and (\ref{hermite3}), follow.
\end{proof}

\begin{Thm}\label{orth}
	The Hermite polynomials are orthogonal by the Gaussian distribution:
	$$\int_{-\infty}^{\infty}H_n(y)H_p(y)g(y)dy=\begin{cases}
	1\quad\mathtt{ \textit{if} }\quad n=p\\
	0\quad\mathtt{ \textit{if} }\quad n\neq p
	\end{cases}$$
\end{Thm}
\begin{proof}
Without loss of generality suppose $n\leq p$, and lets define
$$I_{n,p}=\int_{-\infty}^{\infty}H_n(y)H_p(y)g(y)dy$$
By the substitution of $H_p(y)g(y)$ with $\displaystyle\frac{g^{(p)}(y)}{\sqrt{p!}}$ we get:
$$I_{n,p}=\frac{1}{\sqrt{p!}}\int_{-\infty}^{\infty}H_n(y)g^{(p)}(y)dy.$$
Integrating by parts gives:
$$I_{n,p}=-\frac{1}{\sqrt{p!}}\int_{-\infty}^{\infty}H'_n(y)g^{(p-1)}(y)dy.$$
By using the recurrence formula:
$$I_{n,p}=\frac{\sqrt{n}}{\sqrt{p!}}\int_{-\infty}^{\infty}H_{n-1}(y)g^{(p-1)}(y)dy=\frac{\sqrt{n}}{\sqrt{p!}}I_{n-1,p-1}.$$
Iterating in the same way $n$ times we get:
$$I_{n,p}=\frac{\sqrt{n!}}{\sqrt{p!}}I_{0,p-n}=\frac{\sqrt{n!}}{\sqrt{p!}}\int_{-\infty}^{\infty}H_{0}(y)g^{(p-n)}(y)dy.$$
We get two cases:\sm

\bul if $n=p$, then $I_{n,p}=\displaystyle\frac{\sqrt{n!}}{\sqrt{n!}}\int_{-\infty}^{\infty}g(y)dy=1$\sm

\bul if $n\neq p$, then $I_{n,p}=\displaystyle\frac{\sqrt{n!}}{\sqrt{p!}}\Big[g^{(p-n-1)}(y)\Big]^{+\infty}_{-\infty}=\displaystyle\frac{\sqrt{n!}}{\sqrt{p!}}\Big[H_{(p-n-1)}(y)g(y)\Big]^{+\infty}_{-\infty}=0$
\sm
\end{proof}

If $Y\sim\mathcal {N}(0 ,1)$, the expressions above can be written in probabilistic terms:
$$\mathbb{E}\{H_n(y)H_p(y)\}=\begin{cases}
1\quad\mathtt{ \textit{if} }\quad n=p\\
0\quad\mathtt{ \textit{if} }\quad n\neq p
\end{cases}$$
where the mean of $H_n(y)$ is:
\begin{equation} \label{meanHe}
\mathbb{E}\{H_n(y)\}=\mathbb{E}\{H_n(y)H_0(y)\}=\begin{cases}
1\quad\mathtt{ \textit{if} }\quad n=0\\
0\quad\mathtt{ \textit{if} }\quad n\neq 0
\end{cases}
\end{equation}
and the variance
$$\mathbb{V}ar\{H_n(y)\}=\begin{cases}
1\quad\mathtt{ \textit{if} }\quad n\neq0\\
0\quad\mathtt{ \textit{if} }\quad n=0
\end{cases}$$
\begin{Thm}
	All function $ \phi $ square integrable by $ g $, say that $\int_{-\infty}^{\infty}\phi^2(y)g(y)dy$, can be expanded into Hermite polynomials:
	$$\phi(y)=\sum_{n=0}^{\infty}\phi_n H_n(y)$$
	where the coefficients $ \phi_n $ are given by
	$$\displaystyle\phi_n=\int_{-\infty}^{\infty}\phi(y)H_n(y)g(y)dy,$$
	Also we have
	\begin{equation}\label{sumpol}
	\int_{-\infty}^{\infty}\phi^2(y)g(y)dy=\sum_{n=0}^{\infty}\phi^2_n.
	\end{equation}
\end{Thm}
We will omit this proof of the expansion into Hermite polynomials. The next two formulas can be obtained by using the Theorem (\ref{orth}).

Let $\lambda$ be any value. Let's consider the Taylor expansion of the function $g(y-\lambda)$ in a point $y$:
$$g(y-\lambda)=\sum_{n=0}^{\infty}\frac{(-\lambda)^n}{n!}g^{(n)}(y)=\sum_{n=0}^{\infty}\frac{(-\lambda)^n}{\sqrt{n!}}H_n(y)g(y).$$

Also:
$$\displaystyle g(y-\lambda)=\frac{1}{\sqrt{2\pi}}e^{-\frac{1}{2}(y-\lambda)^2}=g(y)e^{\lambda y-\frac{\lambda^2}{2}}.$$

Equating both sides gives:
\begin{equation}\label{expexp}
\displaystyle e^{\lambda y}=e^{\lambda^2/2}\sum_{n=0}^{\infty}\frac{(-\lambda)^n}{\sqrt{n!}}H_n(y)
\end{equation}

\begin{Prop}
	\begin{equation}\label{exp1}
	H_n(y)=\sum_{q=0}^{n/2}(-1)^{n-q}\frac{\sqrt{n!}2^{-q}}{(n-2q)!q!}y^{n-2q}
	\end{equation}
\end{Prop}
\begin{proof}
	Consider the Fourier transform of the standard Gaussian density:
	\begin{equation}\label{feq1}
		\int_{-\infty}^{\infty}g(t)e^{-iyt}dt=\frac{1}{\sqrt{2\pi}}\int_{-\infty}^{\infty}e^{-\frac{t^2}{2}}e^{-iyt}dt=e^{-\frac{y^2}{2}}
	\end{equation}
The integral can be repeatedly differentiated with respect to $y$. Then we have:
	\begin{equation}\label{feq2}
\frac{d^ne^{-\frac{y^2}{2}}}{dy^n}=\frac{(-i)^n}{\sqrt{2\pi}}\int_{-\infty}^{\infty}e^{-\frac{t^2}{2}}t^ne^{-iyt}dt
\end{equation}
Using the definition (\ref{hermite}) of the Hermite polynomials:
$$H_n(y)=\frac{1}{\sqrt{n!}} \cdot \frac{g^{(n)}(y)}{ g(y)}=\frac{1}{\sqrt{n!}e^{-\frac{y^2}{2}}}\cdot\frac{d^ne^{-\frac{y^2}{2}}}{dy^n},$$
and (\ref{feq2}), Hermite polynomials can be rewritten as:
\begin{equation}\label{feq3}
H_n(y)=\frac{(-i)^ne^{\frac{y^2}{2}}}{\sqrt{2\pi n!}}\int_{-\infty}^{\infty}e^{-\frac{t^2}{2}}t^ne^{-iyt}dt
\end{equation}
Observe that the term $ t^n $ in the integrand suggests that we consider:
$$\sum_{n=0}^{\infty}\frac{H_n(y)r^n}{\sqrt{n!}}=\frac{e^{\frac{y^2}{2}}}{\sqrt{2\pi }}\int_{-\infty}^{\infty}e^{-\frac{t^2}{2}}\Big(\sum_{n=0}^{\infty}\frac{(-irt)^n}{n!}\Big)e^{-iyt}dt=\frac{e^{\frac{y^2}{2}}}{\sqrt{2\pi }}\int_{-\infty}^{\infty}e^{-\frac{t^2}{2}}e^{-i(r+y)t}dt$$
By (\ref{feq1}), this leads to
\begin{equation}\label{feq4}
\sum_{n=0}^{\infty}\frac{H_n(y)r^n}{\sqrt{n!}}=e^{\frac{y^2}{2}}e^{-\frac{(r+y)^2}{2}}=e^{-ry-\frac{r^2}{2}}
\end{equation}
Notice that the exponential can be rewritten as:
$$\sum_{p=0}^{\infty}\frac{(-y)^p}{p!}r^p\sum_{q=0}^{\infty}\frac{{(-1)^{q} r^{2q}}}{\smash{2^q} q!}=\sum_{p=0}^{\infty}\sum_{q=0}^{\infty}\frac{{(-1)^{q} r^{2q+p}}}{\smash{2^q} q!}\frac{(-y)^p}{p!}.$$
By equating both expression and noticing that $n=2q+p$, it leads to
$$H_n(y)=\sum_{q=0}^{n/2}\frac{{(-1)^{q}}\sqrt{n!}(-y)^{n-2q}}{\smash{2^q} q!(n-2q)!}$$
\end{proof}

The Hermite polynomials are not orthonormal for the nonstandard Gaussian
distribution. We will need to compute expressions for shifts and stretches of the polynomials, in terms of non-shifted/non-stretched polynomials.
\begin{Prop}
	\begin{equation}\label{largexp}
		H_p(a+b y)=\sum_{n=0}^{p}\sqrt{C_p^n}(1-b^2)^{(p-n)/2}b^n H_{p-n}\Big(\frac{a}{\sqrt{1-b^2}}\Big)H_n(y)
	\end{equation}
	with $\displaystyle C_p^n=\frac{p!}{n!(p-n)!}$.
\end{Prop}
\begin{proof}
By using (\ref{expexp}), one can write:
\begin{align}
	e^{\lambda(a+by)}&=e^{\lambda^2/2}\sum_{p=0}^{\infty}(-1)^p\frac{\lambda^p}{\sqrt{p!}}H_p(a+by)\label{side1}\\
	e^{\lambda(a+by)}&=e^{\lambda a}e^{(\lambda b)y}\nonumber\\
	&=e^{\lambda a}e^{\lambda^2 b^2/2}\sum_{n=0}^{\infty}\frac{(-\lambda)^n b^n}{\sqrt{n!}}H_n(y)=e^{\lambda a}e^{\lambda^2 (b^2-1+1)/2}\sum_{n=0}^{\infty}\frac{(-\lambda)^n b^n}{\sqrt{n!}}H_n(y)\nonumber\\
	&=\Bigg(\sum_{k=0}^{\infty}\frac{\lambda^k a^k}{k!}\Bigg)\Bigg(e^{\lambda^2/2}\sum_{q=0}^{\infty}\frac{\lambda^{2q}(b^2-1)^q }{\smash{2^q}q!}\Bigg)\Bigg(\sum_{n=0}^{\infty}\frac{(-\lambda)^n b^n}{\sqrt{n!}}H_n(y)\Bigg)\nonumber\\
	&=e^{\lambda^2/2}\sum_{k=0}^{\infty}\sum_{n=0}^{\infty}\sum_{q=0}^{\infty}(-1)^n\frac{a^k b^n \lambda^{n+k+2q}(b^2-1)^q}{k!\smash{2^q}q!\sqrt{n!}}H_n(y).\label{side2}
\end{align}
Equating the coefficients associated with $\lambda$ in (\ref{side1}) and (\ref{side2}), and noticing that $p=n+k+2q$, leads to:
\begin{align}
	H_p(a+by)&=\sum_{n=0}^{p}\sum_{q=0}^{(p-n)/2}(-1)^{p-n}\frac{\sqrt{p!}a^{p-n-2q} b^n (b^2-1)^q}{(p-n-2q)!\smash{2^q}q!\sqrt{n!}}H_n(y)\nonumber\\
	&=\sum_{n=0}\sqrt{\frac{p!}{n!(p-n)!}}a^{p-n} b^n H_n(y)\sum_{q=0}^{(p-n)/2}(-1)^{p-n}\frac{\sqrt{(p-n)!} }{\smash{2^q}(p-n-2q)!q!}\Bigg(-\frac{1-b^2}{a^2}\Bigg)\nonumber
\end{align}
By re-arranging (\ref{exp1}):
$$y^{-k}H_k(y)=\sum_{q=0}^{k/2}\frac{(-1)^{k}\sqrt{k!}}{\smash{2^q} q!(k-2q)!}\Bigg(-\frac{1}{y^2}\Bigg)^q,$$
and plugging $k=p-n$ follows:
\begin{align}
	H_p(a+by)&=\sum_{n=0}\sqrt{\frac{p!}{n!(p-n)!}}a^{p-n} b^n H_n(y)(a/\sqrt{1-b^2})^{-(p-n)}H_{p-n}(a/\sqrt{1-b^2})\nonumber\\
	&=\sum_{n=0}\sqrt{C_p^n}(1-b^2)^{(p-n)/2} b^n H_n(y)H_{p-n}(a/\sqrt{1-b^2})\nonumber
\end{align}
\end{proof}

\section{Application of Hermite Polynomials Expansions in a General Case}
\label{sec:Application}
In this section we show an example, taken from \cite{Emery2005a}, of the application of the results previously shown, when Hermite Polynomials are introduced for the computation of the conditional local mean, appropriate for implementation on any probability distributions presented in earth sciences databases.

The idea, in a first step, is to decompose any transfer function into Hermite Polynomials, in which an expansion up to a degree $P$ is used to give a close approximation to the shape of the function:
\[Z(\textbf{u})=\phi[Y(\textbf{u})]=\sum_{p=0}^{\infty}\phi_p H_n(Y(\textbf{u}))\approx\sum_{p=0}^{P}\phi_p H_p(Y(\textbf{u})).\]
Procedures to calculate the coefficients are derived from the properties already summarized in the previous section, and can be found in \cite{Ortiz2005,wackernagel2013multivariate}. Once the coefficients are computed, $\phi'[Y(\textbf{u})]$ follows directly by using the recurrence relationships:
$$\phi'[Y(\textbf{u})]=-\sum_{p=1}^{\infty}\phi_{p}\sqrt{p} H_{p-1}(Y(\textbf{u}))\approx-\sum_{p=1}^{P}\phi_{p}\sqrt{p} H_{p-1}(Y(\textbf{u})).$$

Then, if we want to calculate $[Z^*(\textbf{u})]$, by plugging the Hermite expansion of the anamorphosis in equation \ref{exp} leads to:
\begin{align*}
[Z^*(\textbf{u})]&=\displaystyle\int_{-\infty}^{\infty}\phi(\sigma_{SK}\cdot\textbf{y}+y^*_{SK}) \cdot g(\textbf{y})d\textbf{y}\\
&=\displaystyle\sum_{p=0}^{\infty}\phi_p\int_{-\infty}^{\infty}H_p(\sigma_{SK}\cdot\textbf{y}+y^*_{SK}) \cdot g(\textbf{y})d\textbf{y}.
\end{align*}
As we see, in order to assess the mean value in the raw variable, Hermite polynomials should be computed in a shifted/stretched version. Thus, we can use equation \ref{largexp}, which leads to:
$$[Z^*(\textbf{u})]=\displaystyle\sum_{p=0}^{\infty}\phi_p\sum_{n=0}^{p}\sqrt{C_p^n}(1-\sigma_{SK}^2)^{(p-n)/2}\sigma_{SK}^n H_{p-n}\Bigg(\frac{y^*_{SK}}{\sqrt{1-\sigma_{SK}^2}}\Bigg)\int_{-\infty}^{\infty}H_n(\textbf{y}) \cdot g(\textbf{y})d\textbf{y}.$$

Since the Hermite polynomials have a zero mean except the first one (equation \ref{meanHe}),
all the terms in the sum vanish except the one associated with $n = 0$, reducing the computation of the mean to:
$$[Z^*(\textbf{u})]=\displaystyle\sum_{p=0}^{\infty}\phi_p(1-\sigma_{SK}^2(\textbf{u}))^{p/2} H_{p}\Bigg(\frac{y^*_{SK}(\textbf{u})}{\sqrt{1-\sigma_{SK}^2(\textbf{u})}}\Bigg).$$

\section{Conclusion}
\label{sec:Conclusion}
The aim of the study was to develop a methodology that allow us the
quantification of the uncertainty in an arbitrary volume conditioned by sampling
data, without the use of the traditional geostatistical simulation, which is time
consuming and hard to manage, specially for large grid sizes. We have accomplished this objective, concluding with a tool suitable for
any probabilistic distribution of the sample data.

For this, we have proceeded with an extension of the traditional multi-Gaussian model, allowing us to obtain the formulations needed, and making explicit the dependence of the uncertainty from the grades, the spatial correlation and conditioning data. Only a Kriging of the Gaussian values is needed to obtain conditional local means and variances.

Also we were able to obtain complete local distributions at any support, in an easy and straightforward way, although with the trade-off cost of resorting to an approximate version of the distribution when several number of points are considered for the change of support.

Future work should consider this last point as a root of study, improving such approximations, if it is the case; the search of simplified versions of the results obtained, by the use of polynomials; and extending the framework to multiple variables, taking advantage of correlated heterotopic data.

\bibliographystyle{acm}
\bibliography{./uassessmnt}

\newpage
\signar
\signjo

\end{document}